\documentclass[11pt]{article}

\IfFileExists{sariel_computer.sty}{\def\sarielComp{1}}{}
\ifx\sarielComp\undefined%
\newcommand{\SarielComp}[1]{}
\newcommand{\NotSarielComp}[1]{#1}%
\else
\newcommand{\SarielComp}[1]{#1}%
\newcommand{\NotSarielComp}[1]{}%
\fi
\newcommand{\IfPrinterVer}[2]{#2}%

\newcommand{\UsePackage}[1]{%
  \IfFileExists{../styles/#1.sty}{%
      \usepackage{../styles/#1}%
   }{%
      \IfFileExists{./styles/#1.sty}{%
         \usepackage{styles/#1}%
      }{%
         \usepackage{#1}%
      }%
   }%
}

\usepackage[cm]{fullpage}%
\usepackage{amsmath}%
\usepackage{amssymb}%
\usepackage{bm}%
\usepackage{xcolor}%
\usepackage{upgreek}%
\usepackage{euscript}%
\usepackage{caption}%

\usepackage{graphicx}
\usepackage{animate}

\usepackage[amsmath,thmmarks]{ntheorem}%
\theoremseparator{.}%

\usepackage{titlesec}%
\titlelabel{\thetitle. }%
\usepackage{xcolor}%
\usepackage{mleftright}%
\usepackage{xspace}%
\usepackage{scalerel}%

\usepackage[inline]{enumitem}

\newlist{compactenumA}{enumerate}{5}%
\setlist[compactenumA]{topsep=0pt,itemsep=-1ex,partopsep=1ex,parsep=1ex,%
   label=(\Alph*)}%

\newlist{compactenumi}{enumerate}{5}%
\setlist[compactenumi]{topsep=0pt,itemsep=-1ex,partopsep=1ex,parsep=1ex,%
   label=(\roman*)}%

\usepackage{ifluatex}
\usepackage{ifxetex}

\ifluatex 
      \usepackage{fontspec}
      \usepackage[utf8]{luainputenc}
\else
       \ifxetex 
          \usepackage{fontspec}
       \else
          \usepackage[T1]{fontenc}
          \usepackage[utf8]{inputenc}
       \fi
\fi
\providecommand{\BibLatexMode}[1]{}
\providecommand{\BibTexMode}[1]{#1}

\usepackage{hyperref}%

\IfPrinterVer{%
   \usepackage{hyperref}%
}{%
   \usepackage{hyperref}%
   \hypersetup{%
      breaklinks,%
      ocgcolorlinks, colorlinks=true,%
      urlcolor=[rgb]{0.25,0.0,0.0},%
      linkcolor=[rgb]{0.5,0.0,0.0},%
      citecolor=[rgb]{0,0.2,0.445},%
      filecolor=[rgb]{0,0,0.4},
      anchorcolor=[rgb]={0.0,0.1,0.2}%
   }
}

\definecolor{blue25}{rgb}{0,0,0.7}
\providecommand{\emphic}[2]{%
   \textcolor{blue25}{%
      \textbf{\emph{#1}}}%
   \index{#2}}
\providecommand{\emphi}[1]{\emphic{#1}{#1}}

\theoremseparator{.}%

\theoremstyle{plain}%
\newtheorem{theorem}{Theorem}[section]

\newtheorem{lemma}[theorem]{Lemma}

\newtheorem{corollary}[theorem]{Corollary}

\newtheorem{observation}[theorem]{Observation}

\theoremstyle{plain}%
\theoremheaderfont{\sf} \theorembodyfont{\upshape}%
\newtheorem*{remark:unnumbered}[theorem]{Remark}%
\newtheorem{remark}[theorem]{Remark}%
\newtheorem{definition}[theorem]{Definition}

\newcommand{\myqedsymbol}{\rule{2mm}{2mm}}

\theoremheaderfont{\em}%
\theorembodyfont{\upshape}%
\theoremstyle{nonumberplain}%
\theoremseparator{}%
\theoremsymbol{\myqedsymbol}%
\newtheorem{proof}{Proof:}%

\newcommand{\atgen}{\symbol{'100}}
\newcommand{\SarielThanks}[1]{\thanks{Department of Computer Science;
      University of Illinois; 201 N. Goodwin Avenue; Urbana, IL,
      61801, USA; {\tt sariel\atgen{}illinois.edu}; {\tt
         \url{http://sarielhp.org/}.} #1}}

\newcommand{\MitchellThanks}[1]{%
   \thanks{%
      Department of Computer Science;
      University of Illinois; 201 N. Goodwin Avenue; Urbana, IL,
      61801, USA; {\tt mfjones2\atgen{}illinois.edu}; {\tt
         \url{http://mfjones2.web.engr.illinois.edu/}.} #1}}

\numberwithin{figure}{section}%
\numberwithin{table}{section}%
\numberwithin{equation}{section}%

\newcommand{\HLink}[2]{\hyperref[#2]{#1~\ref*{#2}}}
\newcommand{\HLinkSuffix}[3]{\hyperref[#2]{#1\ref*{#2}{#3}}}

\newcommand{\thmlab}[1]{{\label{theo:#1}}}
\newcommand{\thmref}[1]{\HLink{Theorem}{theo:#1}}

\newcommand{\corlab}[1]{\label{cor:#1}}
\newcommand{\corref}[1]{\HLink{Corollary}{cor:#1}}%

\newcommand{\seclab}[1]{\label{sec:#1}}
\newcommand{\secref}[1]{\HLink{Section}{sec:#1}}

\newcommand{\remlab}[1]{\label{rem:#1}}
\newcommand{\remref}[1]{\HLink{Remark}{rem:#1}}%

\providecommand{\deflab}[1]{\label{def:#1}}
\newcommand{\defref}[1]{\HLink{Definition}{def:#1}}

\newcommand{\apndlab}[1]{\label{apnd:#1}}
\newcommand{\apndref}[1]{\HLink{Appendix}{apnd:#1}}

\newcommand{\lemlab}[1]{\label{lemma:#1}}
\newcommand{\lemref}[1]{\HLink{Lemma}{lemma:#1}}%

\providecommand{\eqlab}[1]{}%
\renewcommand{\eqlab}[1]{\label{equation:#1}}
\newcommand{\Eqref}[1]{\HLinkSuffix{Eq.~(}{equation:#1}{)}}

\newcommand{\Set}[2]{\left\{ #1 \;\middle\vert\; #2 \right\}}
\newcommand{\pth}[2][\!]{\mleft({#2}\mright)}%
\newcommand{\pbrcx}[1]{\left[ {#1} \right]}%
\newcommand{\Prob}[1]{\mathop{\mathbf{Pr}}\!\pbrcx{#1}}
\newcommand{\Ex}[2][\!]{\mathop{\mathbf{E}}#1\pbrcx{#2}}

\newcommand{\ceil}[1]{\left\lceil {#1} \right\rceil}

\newcommand{\cardin}[1]{\left| {#1} \right|}%

\renewcommand{\th}{th\xspace}

\newcommand{\tldO}{\scalerel*{\widetilde{O}}{j^2}}%

\renewcommand{\Re}{\mathbb{R}}%

\DefineNamedColor{named}{OliveGreen} {cmyk}{0.44,0,0.95,0.90}

\providecommand{\Mh}[1]{#1}

\newcommand{\eps}{\varepsilon}
\newcommand{\epsA}{\Mh{\vartheta}}%
\newcommand{\epsB}{\Mh{\gamma}}%
\newcommand{\epsW}{\Mh{\xi}}%

\newcommand{\body}{\Mh{\mathsf{C}}}
\newcommand{\bodyA}{\Mh{\mathsf{D}}}

\newcommand{\HI}{\Mh{\mathcal{H}^\SuccIters}}%
\newcommand{\constCN}{\mathsf{c}_1}
\newcommand{\cpq}{\Mh{\alpha}}%

\newcommand{\ts}{\hspace{0.6pt}}

\newcommand{\fn}{\Mh{f}}
\newcommand{\LvSetY}[2]{\mathcal{L}_{#1}\pth{#2}}
\newcommand{\DotProdY}[2]{\left\langle #1, #2 \right\rangle}

\newcommand{\hplane}{\Mh{h}}
\newcommand{\cen}{\Mh{\overline{c}}}%
\newcommand{\pnt}{\Mh{p}}%
\newcommand{\pntq}{\Mh{q}}
\newcommand{\PSet}{\ensuremath{\Mh{P}}\xspace}%
\newcommand{\PSetA}{\ensuremath{\Mh{Q}}\xspace}%
\newcommand{\PSetB}{\ensuremath{\Mh{T}}\xspace}%
\newcommand{\Ranges}{\Mh{\mathcal{C}}}%

\newcommand{\VC}{\textsf{VC}\xspace}%
\newcommand{\range}{\Mh{\mathsf{r}}}%

\newcommand{\MeasureChar}{\Mh{\overline{m}}}%
\newcommand{\Measure}[1]{\MeasureChar\pth{#1}}

\newcommand{\sMeasure}[1]{\Mh{\overline{s}}\pth{#1}}

\newcommand{\Wnet}{\Mh{\mathcal{W}}}%
\newcommand{\VVX}[1]{\mathcal{V}\pth{#1}}%

\newcommand{\Net}{\Mh{{S}}}%
\newcommand{\VCProj}[2]{#1_{|#2}}
\newcommand{\VCDim}{\mathrm{dim}^{}_{\VC}}

\newcommand{\Dim}{\Mh{\xi}}%

\providecommand{\Matousek}{Matou{\v s}ek\xspace}
\newcommand{\hrefb}[3][black]{\href{#2}{\color{#1}{#3}}}%

\newcommand{\query}{\Mh{q}}%
\newcommand{\RangeSpace}{\Mh{\mathsf{S}}}%
\newcommand{\GroundSet}{\Mh{\widehat{\mathsf{X}}}}%
\newcommand{\FGroundSet}{\Mh{{\mathsf{X}}}}%
\newcommand{\RangeSet}{\Mh{\EuScript{R}}}%

\newcommand{\BadProb}{\Mh{\varphi}}%

\newcommand{\polylog}{\mathrm{polylog}}
\newcommand{\YY}{\mathcal{Y}}%
\newcommand{\etal}{\textit{et~al.}\xspace}
\newcommand{\hset}{\Mh{{H}}}%

\newcommand{\CHX}[1]{\Mh{\mathcal{CH}}\pth{#1}}
\newcommand{\Sample}{\Mh{S}}%

\newcommand{\rmax}{\Mh{R}}%
\newcommand{\epsilonA}{\widehat{\eps}}%
\newcommand{\cPnt}{\Mh{\overline{\mathrm{c}}}}%

\newcommand{\Arr}{\Mh{\mathop{\mathrm{\EuScript{A}}}}}%
\newcommand{\ArrX}[1]{\Arr\pth{#1}}%
\newcommand{\SampleSize}{\Mh{\mu}}%
\newcommand{\SuccIters}{\Mh{\tau}}%

\newcommand{\SaveContent}[2]{%
   \expandafter\newcommand{#1}{#2}%
}

\newcommand{\RestatementOf}[2]{
   \noindent%
   \textbf{Restatement of #1.}
   {\em #2{}}%
}
\newcommand{\RestatementOfExt}[3]{
   \noindent%
   \textbf{#1 of #2.}
   {\em #3{}}%
}

\begin{document}

\title{Journey to the Center of the Point Set}%

\author{Sariel Har-Peled%
   \SarielThanks{Work on this paper was partially supported by a NSF
      AF awards CCF-1421231, and 
      CCF-1217462.  
   }%
   \and%
   Mitchell Jones%
   \MitchellThanks{}%
}

\date{\today}

\maketitle

\begin{abstract}
    We revisit an algorithm of Clarkson \etal \cite{cemst-acpir-96},
    that computes (roughly) a $1/(4d^2)$-centerpoint in $\tldO(d^9)$
    time, for a point set in $\Re^d$, where $\tldO$ hides
    polylogarithmic terms. We present an improved algorithm that
    computes (roughly) a $1/d^2$-centerpoint with running time
    $\tldO(d^7)$. While the improvements are (arguably) mild, it is
    the first progress on this well known problem in over twenty
    years. The new algorithm is simpler, and the
    running time bound follows by a simple random walk argument, which
    we believe to be of independent interest.  We also present several
    new applications of the improved centerpoint algorithm.
\end{abstract}


\section{Introduction}

\paragraph*{Notation.}
In the following $O(\cdot)$ hides constants that do not depend on the
dimension. $O_d(\cdot)$ hides constants that depend on the dimension
(usually badly -- exponential or doubly exponential, or even
worse). The notation $\tldO(\cdot)$ hides
polylogarithmic factors, where the power of the polylog is independent
of the dimension.

\paragraph*{Computing centerpoints.}

A classical implication of Helly's theorem, is that for any set
$\PSet$ of $n$ points in $\Re^d$, there is a
$1/(d+1)$-centerpoint. Specifically, given a constant
$\alpha \in (0,1)$, a point $\cen \in \Re^d$ is an
$\alpha$-centerpoint if all closed halfspaces containing $\cen$ also
contain at least $\alpha n$ points of $\PSet$. It is currently unknown
if one can compute a $\Omega(1/d)$-centerpoint in polynomial time (in the
dimension). A randomized polynomial time algorithm was presented by
Clarkson \etal \cite{cemst-acpir-96}, that computes (roughly) a
$1/(4d^2)$-centerpoint in $\tldO(d^9)$ time.

\paragraph*{Weak $\eps$-nets.}
Consider the range space $(\PSet, \Ranges)$, where $\PSet$ is a set of
$n$ points in $\Re^d$, and $\Ranges$ is the set of all convex shapes
in $\Re^d$. This range space has infinite \VC dimension, and as such
it is impervious to the standard $\eps$-net constructions.  \emph{Weak
   $\eps$-nets} bypass this issue by using points outside the point
set. While there is significant amount of work on weak $\eps$-nets,
the constructions known are not easy and result in somewhat large
sets. The state of the art is the work by \Matousek and Wagner
\cite{mw-ncwen-04}, which shows a weak $\eps$-net construction of size
$O( \eps^{-d} (\log \eps^{-1})^{O( d^2 \log d)} )$. Rubin recently
improved this result for points in $\Re^2$, proving the existence of 
weak $\eps$-nets of size $O\pth{\eps^{-(1.5 + \gamma)}}$ for 
arbitrarily small $\gamma > 0$ \cite{r-ibwpnp-18}.
Such a weak $\eps$-net $\Wnet$ has the guarantee that any convex set 
$\body$ that contains at least $\eps n$ points of $\PSet$, must 
contain at least one point of $\Wnet$. See \cite{mv-eaen-17} for a 
recent survey of $\eps$-nets and related concepts. See also the recent 
work by Rok and Smorodinsky \cite{rs-wnmp-16} and references therein.

\paragraph*{Basis of weak $\eps$-nets.}
Mustafa and Ray \cite{mr-wenbs-08} showed that one can pick a random
sample $\Sample$ of size $c_d \eps^{-1} \log \eps^{-1}$ from $\PSet$,
and then compute a weak $\eps$-net for $\PSet$ directly from
$\Sample$, showing that the size of the support needed to compute a weak
$\eps$-net is (roughly) the size of a regular
$\eps$-net. Unfortunately, the constant in their sample
$c_d = O\bigl( d^{d} \pth{\log d}^{cd^3 \log d} \bigr)$ is doubly
exponential in the dimension. This constant $c_d$ is related to the
$((d+1)^2, d+1)$-Hadwiger-Debrunner number (the best known upper bounds
on $(p,q)$-Hadwiger-Debrunner numbers can be found in 
\cite{kst-omcig-17, kst-omcig-15}).

In particular, all current results about weak $\eps$-nets suffer from
the ``curse of dimensionality'' and have constants that are at least
doubly exponential in the dimension.


\paragraph*{Our results.}
Let $\PSet$ be a set of $n$ points in $\Re^d$. In addition
to the improved algorithm for computing approximate centerpoints,
we also suggest two alternatives to weak $\eps$-nets as applications,  
and obtain some related results:

\begin{compactenumA}
    \item \textbf{Approximating centerpoints.} %
    We revisit the algorithm of Clarkson \etal \cite{cemst-acpir-96}
    for approximating a centerpoint. We present an improved algorithm,
    which is a variant of their algorithm which runs in $\tldO(d^7)$
    time, and computes roughly a $1/(d+2)^2$-centerpoint. This improves
    both the running time, and the quality of centerpoint
    computed. While the improvements are small (a factor of $d^2$
    roughly in the running time, and a factor of four in the centerpoint
    quality), we believe that the new algorithm is simpler. The
    analysis is cleaner, and is of independent interest.  In
    particular, the analysis uses a random walk argument, which is
    quite uncommon in computational geometry, and (we believe) is of
    independent interest. See \thmref{center:point:compute:epsW}. 
    This is the first improvement of the randomized algorithm of 
    Clarkson \etal \cite{cemst-acpir-96} in over twenty years. Miller 
    and Sheehy also derandomized the algorithm of Clarkson \etal, 
    computing a $\Omega(1/d^2)$-centerpoint in time $n^{O(\log d)}$ 
    \cite{ms-acwp-10}.
  
    \item \textbf{Lowerbounding convex functions.}
    Given a convex function $f$ in $\Re^d$, such that one
    can compute its value and gradient at a point efficiently,
    we present an
    algorithm that computes quickly a realizable lower-bound on
    the value of $f$ over $\PSet$. Formally, the algorithm
    computes a point $\pntq \in \Re^d$, such that
    $f(\pntq) \leq \min_{\pnt \in \PSet} f(\pnt)$. The algorithm
    is somewhat similar in spirit to the ellipsoid algorithm. The
    running time of the algorithm is $\tldO\pth{d^9}$.
    See \thmref{lowerbound:convex:fn}.

    \item \textbf{Functional nets.} 
    Let $\body \subseteq \Re^d$ be a convex body. Suppose we are only
    given access to $\body$ via a separation oracle: given a query 
    point $\query$, the oracle either returns that $\query$ is in 
    $\body$, or alternatively, the oracle returns a hyperplane 
    separating $\query$ and $\body$. We show that a random sample of 
    size
    \begin{align*}
      O \pth{  \eps^{-1} d^3 \log d  \log^3 \eps^{-1} +
      \eps^{-1} \log{\BadProb}^{-1}}%
      =%
      \tldO\pth{ d^3/\eps },
    \end{align*}
    with probability $\geq 1-\BadProb$, can be used to decide if a
    query convex body $\body$ is $\eps$-light.  Formally, the
    algorithm, using only the sample, performs
    $O( d^2 \log \eps^{-1})$ oracle queries -- if any of the query
    points generated stabs $\body$, then $\body$ is considered as
    (potentially) containing more than $\eps n$ points. Alternatively,
    if all the queries missed $\body$, then $\body$ contains less than
    $\eps n$ points of $\PSet$. The query points can be computed in
    polynomial time, and we emphasize that the dependency in the
    running time and sample size are polynomial in $\eps$ and $d$.
    See \thmref{func:net}. As such, this result can be viewed as
    slightly mitigating the curse of dimensionality in the context of
    weak $\eps$-nets.

    \item \textbf{Center nets.} Using the above, one can also
    construct a weak $\eps$-net directly from such a sample -- this
    improves over the result of Mustafa and Ray \cite{mr-wenbs-08} as
    far as the dependency on the dimension is concerned. This is described in
    \apndref{weak:eps:net}.

    Surprisingly, by using ideas from \thmref{func:net} one can get a
    stronger form of a weak $\eps$-net, which we refer to as an
    \emph{$(\eps,\cpq)$-center net}. Here
    $\cpq = \Omega( 1/(d\log \eps^{-1}))$ and one can compute a set $\Wnet$
    of size (roughly) $\tldO_d \bigl( \eps^{-O(d^2)} \bigr)$, such
    that if a convex body $\body$ contains $\geq \eps n$ points of
    $\PSet$, then $\Wnet$ contains a point $\query$ which is an
    $\cpq$-centerpoint of $\body \cap \PSet$. Namely, the net contains
    a point that stabs $\body$ in the ``middle'' as far as the point
    set $\body \cap \PSet$.  See \thmref{center:net}.
\end{compactenumA}

\paragraph*{Paper organization.}

The improved centerpoint approximation algorithm is described in
\secref{approx:centerpoint}. Two applications of the improved
centerpoint algorithm are presented in \secref{algs:oracle:access}.
The construction of center nets is described in \secref{center:nets}.
Background and standard tools used are described in \secref{background}.


\section{Background}
\seclab{background}

\subsection{Ranges spaces, \VC dimension, %
   samples and nets}
\seclab{background:vc}

The following is a quick survey of (standard) known results about
$\eps$-nets, $\eps$-samples, and relative approximations
\cite{h-gaa-11}.
\begin{definition}
    A \emphi{range space} $\RangeSpace$ is a pair
    $(\GroundSet,\RangeSet)$, where $\GroundSet$ is a \emphi{ground
       set} (finite or infinite) and $\RangeSet$ is a (finite or
    infinite) family of subsets of $\GroundSet$. The elements of
    $\GroundSet$ are \emphi{points} and the elements of $\RangeSet$
    are \emphic{ranges}{range}.
\end{definition}

For technical reasons, it will be easier to consider a finite subset
$\FGroundSet \subseteq \GroundSet$ as the underlining ground set.

\begin{definition}
    \deflab{measure}%
    \deflab{s:measure}%
    Let $\RangeSpace = (\GroundSet,\RangeSet)$ be a range space, and
    let $\FGroundSet$ be a finite (fixed) subset of $\GroundSet$. For
    a range $\range \in \RangeSet$, its \emphi{measure} is the
    quantity
    \begin{math}
        \Measure{\range} = {\cardin{\range \cap
              \FGroundSet}}/{\cardin{\FGroundSet}}.
    \end{math}
    %
%
    For a subset $\Net \subseteq \FGroundSet$, its \emphi{estimate} of
    $\Measure{\range}$, for $\range \in \RangeSet$, is the quantity
    \begin{math}
        \sMeasure{\range}%
        =%
        {\cardin{\range \cap \Net}} / {\cardin{\Net}}.
    \end{math}
\end{definition}

\begin{definition}
    \deflab{v:c:dimension}%
    Let $\RangeSpace = (\GroundSet,\RangeSet)$ be a range space. For
    $Y \subseteq \GroundSet$, let
    \begin{math}
        \VCProj{\RangeSet}{Y}%
        =%
        \Set{\range \cap Y}{ \range \in \RangeSet\Bigr.\,}
    \end{math}
    denote the \emphic{projection}{range space!projection} of
    $\RangeSet$ on $Y$. The range space $\RangeSpace$ projected to $Y$
    is $\VCProj{\RangeSpace}{Y} = \pth{Y, \VCProj{\RangeSet}{Y}}$.  If
    $\VCProj{\RangeSet}{Y}$ contains all subsets of $Y$ (i.e., if $Y$
    is finite, we have
    $\cardin{\VCProj{\RangeSet}{Y}} = 2^{\cardin{Y}}$), then $Y$ is
    \emphi{shattered} by $\RangeSet$ (or equivalently $Y$ is shattered
    by $\RangeSpace$).
    
    The \emphi{\VC{}~dimension} of $\RangeSpace$, denoted by
    $\VCDim(\RangeSpace)$, is the maximum cardinality of a shattered
    subset of $\GroundSet$. If there are arbitrarily large shattered
    subsets, then $\VCDim(S) = \infty$.
\end{definition}

\begin{definition}
    Let $\RangeSpace = (\GroundSet,\RangeSet)$ be a range space, and
    let $\FGroundSet$ be a finite subset of $\GroundSet$. For
    $0 \leq \eps \leq 1$, a subset $\Net \subseteq \FGroundSet$ is an
    \emphic{$\eps$-sample}{sample!sample@$\eps$-sample} for
    $\FGroundSet$ if for any range $\range \in \RangeSet$, we have
    \begin{math}
        | \Measure{\range} - \sMeasure{\range} |%
        \leq%
        \eps,
    \end{math}
    see \defref{measure}. Similarly, a set
    $\Net \subseteq \FGroundSet$ is an
    \emphic{$\eps$-net}{net!eps@$\eps$-net|textbf} for $\FGroundSet$
    if for any range $\range \in \RangeSet$, if
    $\Measure{\range} \geq \eps$ (i.e.,
    $\cardin{\range \cap \FGroundSet} \geq \eps
    \cardin{\FGroundSet}$), then $\range$ contains at least one point
    of $\Net$ (i.e., $\range \cap \Net \ne \emptyset$).
    
    A generalization of both concepts is \emph{relative
       approximation}.  Let $p, \epsilonA > 0$ be two fixed constants.
    A \emphi{relative $(p, \epsilonA)$-approximation} is a subset
    $\Net \subseteq \FGroundSet$ that satisfies
    \begin{math}
        (1-\epsilonA) \Measure{\range} \leq %
        \sMeasure{\range} %
        \leq%
        (1+\epsilonA) \Measure{\range},
    \end{math}
    for any $\range \in \Ranges$ such that $\Measure{\range} \geq
    p$. If $\Measure{\range} < p$ then the requirement is that
    \begin{math}
        \cardin{\sMeasure{\range} - \Measure{\range}} \leq \epsilonA
        p.
    \end{math}
\end{definition}

\begin{theorem}[$\eps$-net theorem, \cite{hw-ensrq-87}]
    \thmlab{epsilon:net}%
    Let $(\GroundSet,\RangeSet)$ be a range space of \VC dimension
    $\Dim$, let $\FGroundSet$ be a finite subset of $\GroundSet$, and
    suppose that $0 < \eps \leq 1$ and $\BadProb < 1$. Let $N$ be a
    set obtained by $m$ random independent draws from $\FGroundSet$,
    where
    \begin{math}
        m \geq \max\pth{ \frac{4}{\eps} \lg \frac{4}{\BadProb} \,,\,
           \frac{8\Dim}{\eps}\lg \frac{16}{\eps} }.
    \end{math}
    Then $N$ is an $\eps$-net for $\FGroundSet$ with probability at
    least $1-\BadProb$.
\end{theorem}

The following is a slight strengthening of the result of Vapnik and
Chervonenkis \cite{vc-ucrfe-71} -- see \cite[Theorem 7.13]{h-gaa-11}.

\begin{theorem}[$\eps$-sample theorem]
    \thmlab{epsilon:approximation}%
    Let $\BadProb, \eps> 0$ be parameters and let
    $ \bigl(\GroundSet, \RangeSet \bigr)$ be a range space with \VC
    dimension $\Dim$.  Let $\FGroundSet \subset \GroundSet$ be a
    finite subset.  A sample of size
    \begin{math}
        O\pth{ \eps^{-2}\pth{\Dim + \log{\BadProb}^{-1}}}
    \end{math}
    from $\FGroundSet$ is an $\eps$-sample for
    $\RangeSpace = \pth{\FGroundSet, \RangeSet}$ with probability
    $\geq 1-\BadProb$.
\end{theorem}

\begin{theorem}[{{{\cite{lls-ibscl-01,hs-rag-11}}}}]
    \thmlab{relative}%
    A sample $\Net$ of size
    $ O \pth{ \epsilonA^{-2}p^{-1}\pth{\Dim \log p^{-1} +
          \log{\BadProb}^{-1}}}$ from a range space with \VC dimension
    $\Dim$, is a relative $(p,\epsilonA\ts)$-approximation with
    probability $\geq 1-\BadProb$.
\end{theorem}

The following is a standard statement on the \VC dimension of a range
space formed by mixing several range spaces together (see
\cite{h-gaa-11}).

\begin{lemma}
    \lemlab{mixing:range:spaces}%
    Let
    $\RangeSpace_1 = (\GroundSet_1, \Ranges_1), \ldots, \RangeSpace_k
    = (\GroundSet, \Ranges_k)$ be $k$ range spaces, where all of them
    have the same \VC dimension $\Dim$. Consider the new set of ranges
    \begin{math}
        \widehat{\Ranges}%
        =%
        \Set{r_1 \cap \ldots \cap r_k}{r_1 \in R_1, \ldots, r_k \in
           R_k}.
    \end{math}
    Then the range space
    $\widehat{\RangeSpace} = (\GroundSet, \widehat{\Ranges})$ has \VC
    dimension $O(\Dim k \log k)$.
\end{lemma}


\subsection{Weak $\eps$-nets}

A convex body $\body \subseteq \Re^d$ is \emphi{$\eps$-heavy} (or just
\emph{heavy}) if $\Measure{\body} \geq \eps$ (i.e.,
$\cardin{\body \cap \PSet} \geq \eps \cardin{\PSet}$). Otherwise,
$\body$ is \emphi{$\eps$-light}.

\begin{definition}[Weak $\eps$-net]
    \deflab{weak:e:net}%
    Let $\PSet$ be a set of $n$ points in $\Re^d$. A finite set
    $\Net \subset \Re^d$ is a \emphi{weak $\eps$-net} for $\PSet$ if
    for any convex set $\body$ with $\Measure{\body} \geq \eps$, we
    have $\Net \cap \body \neq \varnothing$.
\end{definition}

Note, that like (regular) $\eps$-nets, weak $\eps$-nets have one-sided
error -- if $\body$ is heavy then the net must stab it, but if $\body$
is light then the net may or may not stab it.

\subsection{Centerpoints}

Given a set \PSet of $n$ points in $\Re^d$, and a constant
$\alpha \in (0,1/(d+1)]$, a point $\cen \in \Re^d$ is an
\emphi{$\alpha$-centerpoint} if for any closed halfspace that contains
$\cen$, the halfspace also contains at least $\alpha n$ points of
$\PSet$.  It is a classical consequence of Helly's theorem that a
$1/(d+1)$-centerpoint always exists. If a point $\cen \in \Re^d$ is a
$1/(d+1)$-centerpoint for $\PSet$, we omit the $1/(d+1)$ and simply
say that $\cen$ is a centerpoint for $\PSet$.


\section{Approximating the centerpoint via %
   Radon's urn}
\seclab{approx:centerpoint}

We revisit the algorithm of Clarkson \etal \cite{cemst-acpir-96} for
approximating a centerpoint. We give a variant of their algorithm, and
present a different (and we believe cleaner) analysis of the
algorithm. In the process we improve the running time from being
roughly $\tldO\pth{d^9}$ to $\tldO \pth{d^7}$, and also improve the
quality of centerpoint computed.

\subsection{Radon's urn}

\subsubsection{Setup}

In the \emphi{Radon's urn} game there are $r$ red balls, and $b = n-r$
blue balls in an urn, and there is a parameter $t$. An iteration of
the game goes as follows: 
\begin{compactenumA}
    \item The player picks a random ball, marks it for deletion, and
    returns it to the urn.

    \item The player picks a sample $S$ of $t$ balls (with replacement
    -- which implies that we might have several copies of the same
    ball in the sample) from the urn.

    \item If at least two of the balls in the sample $S$ are red, the
    player inserts a new red ball into the urn. Otherwise, the player
    inserts a new blue ball.

    \item The player returns the sample to the urn.

    \item Finally, the player removes the ball marked for deletion
    from the urn.
\end{compactenumA}

Note that in each stage of the game, the number of balls in the urn
remains the same. We are interested in how many rounds of the game one
has to play till there are no red balls in the urn with high
probability. Here, the initial value of $r$ (i.e., $r_0$) is going to
be relatively small compared to $n$.

\subsection{Analysis}

Let $P(r)$ be the probability of picking two or more red balls into
the sample, assuming that there are $r$ red balls in the urn. We have
that
\begin{equation*} 
  P(r)%
  =%
  \sum_{i=2}^{t} \binom{t}{i} \pth{\frac{r}{n}}^i
  \pth{1-\frac{r}{n}}^{t-i}%
  \leq%
  \binom{t}{2} \pth{\frac{r}{n}}^2
  \leq%
  \frac{t^2}{2} \pth{\frac{r}{n}}^2. %
\end{equation*}
(For the proof of the first inequality, see \apndref{silly}.)  Note,
that $P(r) \leq 1/8$ if $n \geq 2tr$.  Let $P_+(r)$ be the probability
that the number of red balls increased in this iteration. For this to
happen, at least two red balls had to be in the sample, and the
deleted ball must be blue.  Let $P_-(r)$ be the probability that the
number of red balls decreases -- the player needs to pick strictly
less than two red balls in the sample, and delete a red ball. This
implies
\begin{equation*}
  P_+(r) = P(r) (1-r/n) \leq P(r) %
  \qquad \text{ and } \qquad%
  P_{-}(r) = (1-P(r))(r/n).
\end{equation*}
The probability for a change in the number of red balls at this
iteration is
\begin{equation*}
  P_{\pm}(r)%
  =%
  P_+(r) + P_-(r)%
  =%
  P(r) (1-r/n) + (1-P(r))(r/n)%
  =%
  (1-2r/n)P(r) + r/n.
\end{equation*}

\SaveContent{\LemmaGoodRange}{%
   Let $\epsW \in (0, 1/4)$ be a parameter.  If $r \leq \rmax$, then
  \begin{align*}
      P_+(r) \leq \frac{1/2-\epsW}{1/2+\epsW} P_{-}(r), \qquad \text{ where }\qquad
      \rmax =
      (1 - 2\epsW)\frac{n}{t^2}.
  \end{align*}
}%
\begin{lemma}(Proof in \apndref{good:range:proof}.) %
    \lemlab{good:range}%
    \LemmaGoodRange
\end{lemma}

To simplify exposition, we choose $\epsW = 1/6$ so that $P_+(r)
\leq P_-(r)/2$. In the remaining analysis, the proofs involving tedious
calculations are given in \apndref{missing:proofs:approx:centerpoint},
with the modification that the statement of the lemmas and proofs
involve the parameter $\epsW$.

\SaveContent{\LemmaGoodRangeProof}{
  We solve
  $P(r) \leq \frac{1-2\epsW}{1+2\epsW} P_{-}(r) \leq (1-2\epsW)
  P_{-}(r)$, which is equivalent to
  \begin{equation*}
      P(r) \leq
      \bigl(1-P(r)\bigr)\frac{r}{n}  (1 -2 \epsW)
      \iff%
      \frac{n}{ (1 - 2\epsW)r} P(r) \leq 1-P(r)
      \iff%
      P(r)
      \leq \frac{(1 - 2\epsW)r}{(1 - 2\epsW)r + n}.
  \end{equation*}
  
  Now the last inequality holds if
  \begin{math}
      \frac{t^2r^2}{2n^2} \leq \frac{(1 - 2\epsW)r}{2n}
  \end{math}
  since
  \begin{equation*}
      P(r) \leq
      \frac{t^2r^2}{2n^2}
      \leq
      \frac{(1 - 2\epsW)r}{2n}
      \leq
      \frac{(1 - 2\epsW)r}{(1 - 2\epsW)r + n}.
  \end{equation*}
  Namely, the desired condition holds when 
  \begin{math}
      r \leq {(1 - 2\epsW)n}/{t^2}.
  \end{math}
}

\paragraph*{What are we analyzing?}
The value $\rmax/n$ is an upper bound on the ratio of red balls that the urn can
have and still, with good probability, end with zero red balls at the end of the
game. If this ratio is violated anytime during the game, then the urn might end
up consisting of only red balls. We want to start the game with an urn initially
having close to $\rmax$ red balls, but still end up with an entirely blue urn
with sufficiently high probability.

\paragraph*{The question.}
Let $\epsA \in (0,1)$ and $r_0 = (1-\epsA)\rmax$ be the number of red balls in
the urn at the start of the game (note that both $r_0$ and $\rmax$ are functions
of $n$ and $t$). Let $\BadProb > 0$ be a parameter. The key question is the
following: How large does $n$ need to be so that if we start with $r_0$ red
balls (and $n-r_0$ blue balls), the game ends with all balls being blue with
probability $\geq 1 - \BadProb$?


\paragraph*{The game as a random walk.}
An iteration of the game where the number of red balls changes is an
\emphi{effective} iteration. Considering only the effective iterations, this can
be interpreted as a random walk starting at $X_0 = (1-\epsA)\rmax$ and at every
iteration either decreasing the value by one with probability at least $2/3$,
and increasing the value with probability at most $1/3$ (since $P_+(r) \leq
P_-(r)/2$).  This walk ends when either it reaches $0$ or $\rmax$. If the walks
reaches $\rmax$, then the process fails. Otherwise if the walk reaches $0$, then
there are no red balls in the urn, as desired.

\subsubsection{Analyzing the  related walk}
Consider the random walk that starts at $Y_0 = (1-\epsA)\rmax$. In the
$i$\th iteration, $Y_i = Y_{i-1} -1$ with probability $2/3$ and
$Y_i = Y_{i-1} +1$ with probability $1/3$. Let
$\YY = Y_1, Y_2, \ldots$ be the resulting random walk which
stochastically dominates the walk $X_0, X_1, \ldots$. This walk is
strongly biased towards going to $0$, and as such it does not hang
around too long before moving on, as testified by the following lemma.

\SaveContent{\LemmaNumberVisits}
{
  Let $I$ be any integer number, and let $\BadProb > 0$ be a
  parameter. The number of times the random walk $\YY$ visits $I$ is
  at most
  \begin{math}
      \displaystyle%
      \frac{1}{4\epsW^2 }\ln \frac{1}{4 \epsW^2 \BadProb} 
  \end{math}
  times, and this holds with probability $\geq 1-\BadProb$.
}

\begin{lemma}(Proof in \apndref{number:visits:proof}.) %
    \lemlab{number:visits}%
  Let $I$ be any integer number, and let $\BadProb > 0$ be a
  parameter. The number of times the random walk $\YY$ visits $I$ is
  at most
  \begin{math}
      \displaystyle%
      9\ln\pth{9/\BadProb}
  \end{math}
  times, and this holds with probability $\geq 1-\BadProb$.
\end{lemma}

\SaveContent{\LemmaNumberVisitsProof}
{
    Let $\tau$ be the first index such that $Y_\tau = I$. The
    probability that $Y_{\tau+2i} = I$ is
    \begin{align*}
      p_{\tau+i} = 
      \binom{ 2i}{ i} \pth{\frac{1}{2} - \epsW}^{i} \pth{\frac{1}{2} +
      \epsW}^{i}%
      \leq%
      2^{2i} \pth{\frac{1}{4} -\epsW^2}^i
      =%
      \pth{1 -4\epsW^2}^i.
    \end{align*}
    The probability that the walk would visit $I$ again after time
    $\tau+u$ is bounded by 
    \begin{math}
        \nu(u) = \sum_{i=u}^\infty p_{\tau+i} \leq {p_{\tau+u}}/
        \pth{4\epsW^2}.
    \end{math}
    We want to choose $u$ so that $\nu(u) \leq \BadProb$.
    In particular, 
    $\nu(u) \leq (1 - 4\epsW^2)^u/(4 \epsW^2)
            \leq \exp\pth{ - 4 \epsW^2 u }/(4\epsW^2) \leq \BadProb$
    for 
    \begin{math}
        u \geq \frac{1}{4\epsW^2 }\ln \frac{1}{4 \epsW^2 \BadProb}.
    \end{math}
    Thus the walk might visit $I$ during the first $u$ iterations, but
    with probability $\geq 1-\BadProb$ it never visits it again. We
    conclude, that with probability $\geq 1-\BadProb$, the walk visits
    the value $I$ at most $u$ times.
}

We next bound the probability that the walk fails.

\SaveContent{\LemmaHBound}{
  Let $\BadProb > 0$ be a parameter.  If
  \begin{math}
      \displaystyle%
      \rmax \geq \frac{1}{ 2 \epsW \epsA} \ln
      \frac{1}{4\epsW^2\BadProb}
  \end{math}
  then the probability that the random walk ever visits $\rmax$ (and
  thus fails) is bounded by $\BadProb$.
}

\begin{lemma}(Proof in \apndref{H:bound:proof}.) %
    \lemlab{H:bound}%
    Let $\BadProb > 0$ be a parameter. If
    \begin{math}
        \displaystyle%
        \rmax \geq \frac{3}{\epsA} \ln \frac{9}{\BadProb}
    \end{math}
    then the probability that the random walk ever visits $\rmax$ (and
    thus fails) is bounded by $\BadProb$.
\end{lemma}

\SaveContent{\LemmaHBoundProof}{%
   Since the walk starts at $(1-\epsA)\rmax$, the first time it can
   arrive to $\rmax$ is at time $\epsA \rmax$. In particular, the
   probability that $Y_{\epsA \rmax + 2i} = \rmax$ is
   \begin{align*}
      p_i'%
      &=%
        \binom{\epsA \rmax + 2i}{\epsA \rmax + i}
        \pth{\frac{1}{2} - \epsW}^{\epsA\rmax+i}
        \pth{\frac{1}{2} + \epsW}^{i}%
        \leq %
        2^ {\epsA \rmax + 2i}\pth{\frac{1}{2} - \epsW }^{\epsA \rmax+i}
        \pth{\frac{1}{2} + \epsW}^{i}%
        =%
        \pth{1 - 2\epsW}^{\epsA \rmax+i}
        \pth{1 + 2\epsW}^{i}%
      \\&%
      = 
      \pth{1 - 2\epsW}^{\epsA \rmax}
      \pth{1 - 4\epsW^2}^{i}.%
    \end{align*}
    It follows that the probability of failure is bounded by
    \begin{math}
        p = \sum_{i=0}^\infty p_i' \leq {\pth{1 - 2\epsW}^{\epsA
              \rmax}} /\pth{4\epsW^2} 
            \leq \exp(-2\epsW\epsA \rmax)/(4\epsW^2).
    \end{math}
    In particular, we require that
    \begin{math}
        p \leq \exp\pth{ - 2\epsW \epsA
           \rmax}/ \pth{4\epsW^2} \leq \BadProb,
    \end{math}
    which holds for
    \begin{math}
        \displaystyle%
        \rmax \geq \frac{1}{ 2 \epsW \epsA} \ln
        \frac{1}{4\epsW^2\BadProb}.
    \end{math}
}

\subsubsection{Back to the urn}

The number of red balls in the urn is stochastically dominated by the
random walk above.  The challenge is that the number of iterations one
has to play before an effective iteration happens (thus, corresponding
to one step of the above walk), depends on the number of red balls,
and behaves like the coupon collector problem. Specifically, if there
are $r \leq \rmax$ red balls in the urn, then the probability for an
effective step is $P_\pm(r) \geq (1-P(r))(r/n) \geq r/2n$, as
$P(r) \leq 1/2$.  This implies that, in expectation, one has to wait
at most $2n/r$ iterations before an effective iteration happens.

\SaveContent{\LemmaStuckOnR}{ Let $\BadProb > 0$ and
   $\epsW \in (0,1/4)$ be parameters.  For any value $r \leq \rmax$,
   the urn spends at most
   $s(n,r) = O\pth{ (n/r) \epsW^{-2} \ln (\epsW^{-1} \BadProb^{-1})}$
   regular iterations, throughout the game, having exactly $r$ balls
   in it, with probability $\geq 1-\BadProb$.  }

\SaveContent{\LemmaStuckOnRProof}{
    Consider the iteration that the urn has $r$ red balls. As discussed,
    one has to wait in expectation at most $2n/r$ steps before
    the number of red balls changes. Now by Markov's inequality,
    the probability this takes more than $4n/r$ steps is at most half.
    Namely, if the number of iterations till an effective iteration is
    denoted by $\Delta$, then let $X = \ceil{\Delta/(4n/r)}$ be the
    number of blocks we had to wait till an effective iteration
    happens. It is easy to verify that $X$ has a geometric
    distribution, with $p=1/2$.

    Hence, every run of iterations with the urn having $r$ red
    balls, corresponds to a random variable of such blocks. Let
    $X_1, X_2, \ldots, X_m$ be the number of such blocks, in the first
    $m$ such runs, where we can take
    $m =O\pth{\epsW^{-2} \ln (\epsW^{-2} \BadProb^{-1})}$
    by \lemref{number:visits}.  It is easy to verify that
    $\Ex{\sum_i X_i} = 2m$. Chernoff type inequalities for the sum of
    such geometric variables show that this sum is smaller than, say,
    $8m$ with probability $2^{-3m}$ (see \corref{geometric:sum}
    below). As for the number of iterations the urn spends having
    exactly $r$ balls in it -- this is bounded by $(2n/r) \sum_i X_i$.
}

\begin{lemma}(Proof in \apndref{stuck:on:r:proof}.) %
    \lemlab{stuck:on:r}%
    Let $\BadProb > 0$ be the probability of failure. For
    any value $r \leq \rmax$, the urn spends at most
    $O\pth{ (n/r) \log ( \BadProb^{-1})}$ regular
    iterations, throughout the game, having exactly $r$ balls in it,
    with probability $\geq 1-\BadProb$.
\end{lemma}

\SaveContent{\LemmaIterations}{
    Let $\BadProb > 0$ and $\epsA \in (0,1)$ be parameters, and
    assume that
    \begin{math}
        \displaystyle%
        n = \Omega\pth{ \frac{t^2}{ \epsW \epsA} \ln
           \frac{1}{\epsW\BadProb}}.
    \end{math}
    The
    total number of regular iterations one has to play till the urn
    contains only blue balls, is \break
    $O\pth{ n \epsW^{-2} \log n \log (n \epsW^{-1} \BadProb^{-1})}$,
    and this bound holds with probability $\geq 1- \BadProb$.
}

\SaveContent{\LemmaIterationsProof}{
    The bound on the number of steps follows readily by summing up the
    bound of \lemref{stuck:on:r}, for $r=1, \ldots,
    \rmax$. Specifically, we apply this lemma with failure probability
    $\BadProb/2n$, to bound the sum $\sum_{r=1}^\rmax s(n,r)$, which
    is bounded by the stated bound.  The probability of failure is at
    most $\rmax \BadProb /2n \leq \BadProb/2$.

    The other reason for a failure is that the urn reaches a state
    where it has $\rmax$ red balls. In particular, using
    \lemref{H:bound} to bound this probability by $\BadProb/2$,
    requires that
    \begin{math}
        \displaystyle%
        \rmax \geq \frac{1}{ 2 \epsW \epsA} \ln
        \frac{1}{4\epsW^2\BadProb/2}.
    \end{math}
    By the value specified in \lemref{good:range}, this requires
    \begin{math}
        (1 - 2\epsW)\frac{n}{t^2}
        \geq  \frac{1}{ 2 \epsW \epsA} \ln
        \frac{1}{4\epsW^2\BadProb/2},
    \end{math}
    which holds for the value of $n$ stated.
}

\begin{lemma}(Proof in \apndref{iterations:proof}.) %
    \lemlab{iterations}%
    Let $\BadProb > 0$ and $\epsA \in (0,1)$ be parameters, and
    assume that
    \begin{math}
        \displaystyle%
        n = \Omega\pth{ \frac{t^2}{ \epsA} \ln
           \frac{1}{\BadProb}}.
    \end{math}
    The
    total number of regular iterations one has to play till the urn
    contains only blue balls, is
    $O\pth{ n \log n \log (n \BadProb^{-1})}$,
    and this bound holds with probability $\geq 1- \BadProb$.
\end{lemma}

\subsection{Approximating a centerpoint}

\subsubsection{The algorithm}

Before describing the algorithm, we need the following well known
facts \cite{m-ldg-02}:
\begin{compactenumA}
    \item \textbf{Radon's theorem}: Given a set $\PSetB$ of $d+2$
    points in $\Re^d$, one can partition $T$ into two non-empty sets
    $\PSetB_1, \PSetB_2$, such that
    $\CHX{\PSetB_1} \cap \CHX{\PSetB_2} \neq \emptyset$. A point
    $\pnt \in \CHX{\PSetB_1} \cap \CHX{\PSetB_2}$ is a \emphi{Radon
       point}.
       
    \item Computing a Radon point can be done by solving a system of
    $d+1$ linear equalities in $d+2$ variables. This can be
    completed in $O(d^3)$ time using Gaussian elimination.

    \item A Radon point is a $2/(d+2)$-centerpoint of $\PSetB$.

    \item Let $h^+$ be a halfspace containing only one point of
    $\PSetB$. Then, a Radon point $\pnt$ of $\PSetB$ is contained in
    $\Re^d \setminus h^+$ \cite{cemst-acpir-96}.
\end{compactenumA}

\paragraph*{The algorithm in detail.}
Let $\PSet$ be a set of $n$ points in $\Re^d$ for which we would like
to approximate its centerpoint. To this end, let $\PSetA$ be initially
$\PSet$. In each iteration the algorithm randomly picks $d+2$
points (with repetition) from $\PSetA$, computes their Radon point,
randomly deletes any point of $\PSetA$, and inserts the new Radon point
into the point set $\PSetA$. The claim is that after a sufficient number
of iterations, any point of $\PSetA$ is a $f(d)$-centerpoint of
$\PSet$, where $f(d) =\Theta(1/d^2)$ (its exact value is specified
below in \Eqref{f:d}).

\begin{remark}
    The algorithm above is a variant of the algorithm of Clarkson
    \etal \cite{cemst-acpir-96}. Their algorithm worked in rounds, in
    each round generating $n$ new Radon points, and then replacing the
    point set with this new set, repeating this sufficient number of
    times. Our algorithm on the other hand is a ``continuous''
    process.
\end{remark}

\paragraph*{Intuition.}
A Radon point is a decent center for the points defining it. Visually,
the above algorithm causes the points to slowly migrate towards the
center region of the original point set.

To see why this is true, pick an arbitrary halfspace $h^+$ that
contains exactly $f(d) n$ points of $\PSet$. In each iteration, only
if we picked two (or more) points that are in $\PSetA \cap h^+$, the
new point might be in $h^+$. Observe that we are in the setting of
Radon's urn with $t=d+2$. Indeed, color all the points inside $h^+$ as
red, and all the points outside as blue. To apply the Radon's urn
analysis above, we require that $(1-\epsA)\rmax = f(d)n$, which by
the choice of $\rmax = (1 - 2\epsW)n/t^2$ 
in \lemref{good:range} (and recalling $\epsW = 1/6$) 
implies that
\begin{equation}
    (1-\epsA)\frac{2n}{3(d+2)^2} = f(d)n%
    \iff%
    f(d) = \frac{2(1-\epsA)}{3(d+2)^2}
    \geq \frac{1 - \epsA}{2(d+2)^2}
    \eqlab{f:d}
\end{equation}
where $\epsA \in (0,1)$.  We
can now apply the Radon's urn analysis to argue that after sufficient
number of iterations, all the points of $\PSetA$ are outside
$h^+$. Naturally, we need to apply this analysis to all halfspaces.

\subsubsection{Analysis}
\seclab{centerpoint:analysis}

Consider all half-spaces that might be of interest. To this end,
consider any hyperplane passing through $d$ points of $\PSet$, and
translate it so that it contains on one of its sides exactly $f(d) n$
points (naturally, the are two such translations). Each such
hyperplane thus defines two natural halfspaces. Let $\hset$ be the
resulting set of halfspaces. Observe that
$\cardin{\hset} \leq 2\binom{n}{d} \leq 2(ne/d)^d$. If $\PSetA$ does
not contain any point in any of the halfspaces of $H$ then all its
points are $f(d)$-centerpoints. In particular, one can think about
this as playing $\cardin{\hset}$ parallel Radon's urn games.  We want
the algorithm to succeed with probability $\geq 1 - \BadProb$.
Setting the probability of success for each halfplane of $\hset$ to be
$\BadProb/\cardin{\hset}$, and by \lemref{iterations}, we have that
all of these halfspaces are empty after playing
\begin{equation*}
    O\pth{\bigl. n \log n \log( n |\hset| \BadProb^{-1})}
    =%
    O\pth{\bigl. d n \log n \log( n / \BadProb)}  
\end{equation*}
iterations, with probability of success being
$1 - \cardin{\hset}(\BadProb/\cardin{\hset} ) = 1 - \BadProb$ by the
union bound.  Using \lemref{iterations} requires that
\begin{math}
    n%
    =%
    \Omega \bigl( t^2 \epsA^{-1} \ln (|\hset| / \BadProb)
    \bigr)%
    =%
    \Omega\bigl( \epsA^{-1}d^3\ln n +
    \epsA^{-1}d^2 \ln \BadProb^{-1}\bigr)
\end{math}
which holds for
$n = \Omega(\epsA^{-1}d^3 \ln d + \epsA^{-1}d^2
\ln \BadProb^{-1})$.

Using the fact that computing a Radon point for $d+2$ points in 
$\Re^d$ can be done in $O(d^3)$ time, we obtain the following result.

\SaveContent{\LemmaCenter}{
    Let $\BadProb > 0$ and $\epsB \in (0,1)$ be parameters, and let
    $\PSet$ be a set of
    $n = \Omega(\epsB^{-2}d^3 \ln d + \epsB^{-2}d^2 \ln
    \BadProb^{-1})$ points in $\Re^d$.  Let
    \begin{math}
        \cpq = \frac{1-\epsB}{(d+2)^2}.
    \end{math}
    Then, one can compute a $\cpq$-centerpoint of $\PSet$ via a
    randomized algorithm. The running time of the randomized algorithm
    is
    \begin{math}
        O\pth{ \epsB^{-2} d^3 \cdot d n \log n \log( n \epsB^{-1} \BadProb^{-1})}%
        =%
        O(\epsB^{-2} d^4 n \log n \log( n \epsB^{-1} \BadProb^{-1})),
    \end{math}
    and it succeeds with probability
    $\geq 1 - \BadProb$.
}

\begin{lemma}(A proof of a more general version is in
    \apndref{center:proof}.) %
    \lemlab{center}%
    Let $\BadProb > 0$ and $\epsA \in (0,1)$ be parameters, and let
    $\PSet$ be a set of
    $n = \Omega(\epsA^{-1}d^3 \ln d + \epsA^{-1}d^2 \ln
    \BadProb^{-1})$ points in $\Re^d$.  Let
    \begin{math}
        \cpq = \frac{1-\epsA}{2(d+2)^2}.
    \end{math}
    Then, one can compute a $\cpq$-centerpoint of $\PSet$ via a
    randomized algorithm. The running time of the randomized algorithm
    is
    \begin{math}
        O\pth{d^3 \cdot d n \log n \log( n / \BadProb)}%
        =%
        O(d^4 n \log n \log( n/\BadProb)),
    \end{math}
    and it succeeds with probability
    $\geq 1 - \BadProb$.
\end{lemma}

\SaveContent{\ThmCenterPointCompute} {%
   Given a set $\PSet$ of $n$ points in $\Re^d$, a parameter
   $\BadProb$, and a constant $\epsB \in (0,1)$, one can compute a
   $\frac{1-\epsB}{(d+2)^2}$-centerpoint of $\PSet$. The running time
   of the algorithm is
   \begin{equation*}
       O\pth{ \epsB^{-4} d^7 \log^3 \!d \, \log^3 \pth{ \epsB^{-1}
         \BadProb^{-1} }},
   \end{equation*}
   and it succeeds with probability $\geq 1-\BadProb\bigr.$.%
}

\begin{theorem}
    \thmlab{center:point:compute}
    Given a set $\PSet$ of $n$ points in $\Re^d$, a parameter
    $\BadProb$, and a constant $\epsA \in (0,1)$, one can compute a
    $\frac{1-\epsA}{2(d+2)^2}$-centerpoint of $\PSet$. The running time
    of the algorithm is
    $O\pth{ \epsA^{-3} d^7 \log^3 d \log^3 \BadProb^{-1} }$, and it
    succeeds with probability $\geq 1-\BadProb\bigr.$.
\end{theorem}

\begin{proof}
    The idea is to
    pick a random sample $\Net$ from $\PSet$ that is a
    $(\rho,\epsA/8)$-relative approximation for halfspaces, where
    $\rho = 1/(10d^2)$. This range space has \VC dimension $d+1$, and
    by \thmref{relative}, a sample of size
    \begin{align*}
      \SampleSize %
      =%
      O\bigl(\rho^{-1} \epsA^{-2}(d \log \rho^{-1} + \log
      \BadProb^{-1})\bigr)%
      =%
      O\bigl(d^2\epsA^{-2} ( d \log d + \log \BadProb^{-1})\bigr)
    \end{align*}
    is a $(\rho,\epsA/8)$-relative approximation.

    Running the algorithm of \lemref{center} on $\Net$ with
    $\epsA/8$ yields a $\tau$-centerpoint $\cPnt$ of $\Net$, where
    \begin{math}
        \tau= \frac{1-\epsA/8}{2(d+2)^2\bigr.}\geq \rho
    \end{math}
    for $d \geq 2$.
    By the relative approximation property, this is a
    $(1\pm \epsA/8)\tau$-centerpoint of $\PSet$. Therefore
    $\cPnt$ is an $\cpq$-centerpoint for $\PSet$, where
    \begin{align*}
      \cpq %
      =%
      (1-\epsA/8) \tau %
      =%
      \frac{(1-\epsA/8)^2}{(d+2)^2}
      \geq%
      \frac{1-\epsA}{(d+2)^2}.
    \end{align*}
    By \lemref{center}, the running time of the resulting algorithm is
    \begin{align*}
      O(d^4 \SampleSize \log \SampleSize \log(
      \SampleSize/\BadProb))%
      =%
      O\bigl( \epsA^{-2} \log^2 \epsA^{-1} d^7 \log^3 d \log^3
      \BadProb^{-1} \bigr)
      =%
      O\bigl( \epsA^{-3} d^7 \log^3 d \log^3
      \BadProb^{-1} \bigr).
    \end{align*}
\end{proof}

Now, one can repeat the above calculations with the parameter $\epsW$.
Intuitively, as $\epsW$ approaches zero, the random walk becomes less 
unbalanced since it will move left with probability $1/2 + \epsW$ and
right with probability $1/2 - \epsW$. Because of this, there is an
increased chance that the random walk will reach $\rmax$ (and thus fail).
In order to preserve that the random walk succeeds with
probability at least $1 - \varphi$, the sample size $n$ must depend
on the parameter $\epsW$. In fact, the parameter $\epsW$ allows 
us to compute centerpoints with quality arbitrarily close to 
$1/(d+2)^2$.

\begin{theorem}(Proof in \apndref{center:point:compute:proof:eps:W}.)
    \thmlab{center:point:compute:epsW}%
    \ThmCenterPointCompute
\end{theorem}

\begin{remark}
    The above compares favorably to the result of \cite[Corollary
    3]{cemst-acpir-96} -- they get a running time of
    $O( d^9 \log d + d^8\log^2 \BadProb^{-1})$, which is slower
    by roughly a factor of $d^2$, and computes a
    \begin{math}
        \frac{1}{4.08 (d+2)(d+1)}
    \end{math}%
    -centerpoint of $\PSet$ -- the quality of the centerpoint is
    roughly worse by a factor of four.
\end{remark}


\section{Application I: Algorithms with oracle access}
\seclab{algs:oracle:access}
In this section we discuss two applications of the improved centerpoint
algorithm. Both applications revolve around the idea of \emph{oracle access}. In
the first application, we are interested in lower bounding a convex function given
an oracle to compute it's gradient. In the second, we utilize centerpoints in
order to determine whether a given convex body $\body$ is $\eps$-heavy using a
separation oracle.

\subsection{Lowerbounding a function with a gradient oracle}
\begin{definition}
Let $\fn : \Re^d \to \Re$ be a convex function. For a number $c \in \Re$, 
define the \emphi{level set of $\fn$} as 
$\LvSetY{\fn}{c} = \Set{\pnt \in \Re^d}{\fn(\pnt)
\leq c}$. 
If $\fn$ is a convex function, then $\LvSetY{\fn}{c}$ 
is a convex set for all $c \in \Re$.
\end{definition}
\begin{definition}
Let $\fn : \Re^d \to \Re$ be a convex (and possibly non-differentiable)
function. For a point $\pnt \in \Re^d$, a vector $v \in \Re^d$ is a
\emphi{subgradient} of $\fn$ at $\pnt$ if for all 
$\pntq \in \Re^d$, $\fn(\pntq) \geq \fn(\pnt) + \DotProdY{v}{\pntq - \pnt}$.  
The \emphi{subdifferential} of $\fn$ at $\pnt \in \Re^d$, denoted by 
$\partial \fn(\pnt)$, is the set of all subgradients $v \in \Re^d$ of 
$\fn$ at $\pnt$. When the domain of $\fn$ is $\Re^d$ and $f$ is convex, 
then $\partial \fn(\pnt)$ is a non-empty set for all 
$\pnt \in \Re^d$.\footnote{Indeed, consider the epigraph of the function 
$\fn$, $S = \Set{(x,c)}{x \in \Re^d, c \in \Re, \fn(x) \leq c} 
\subseteq \Re^{d+1}$. Observe that $S$ is a convex set. For a given point
$\pnt \in \Re^d$, by the supporting hyperplane theorem, there is some 
hyperplane $h$ tangent to $S$ at the point $(\pnt,\fn(\pnt)) \in \Re^{d+1}$.
The normal to this tangent hyperplane $h$ is one possible subgradient of 
$\fn$ at $\pnt$.}
\end{definition}

\begin{theorem}
  \thmlab{lowerbound:convex:fn}
  Let $\fn : \Re^d \to \Re$ be a convex (possibly non-differentiable)
  function and $\PSet$ a set of $n$ points in $\Re^d$. 
  Assume that one has access to an oracle which given $\pnt \in
  \Re^d$ returns an arbitrary element in the subdifferential 
  $\partial \fn(\pnt)$. 
  With $O(d^2 \log n)$ queries to the oracle, one can
  compute a point $\pntq \in \Re^d$ (not necessarily in $\PSet$) such
  that $\fn(\pntq) \leq \min_{\pnt \in \PSet} \fn(\pnt)$.
\end{theorem}
\begin{proof}
  Let $\PSet_1 = \PSet$, and $\PSet_i \subseteq \PSet$ denote the set of remaining
  points at the beginning of the $i$th iteration. In iteration $i$, for some
  constant $c > 0$, compute a $(c/d^2)$-centerpoint $\cen_i$ of $\PSet_i$ using
  \thmref{center:point:compute} in time $O(d^7 \log^3 d)$ with success
  probability 1/2. Define $\body_i = \LvSetY{\fn}{\fn(\cen_i)}$. We now use the
  oracle to obtain subgradient vector $v \in \partial \fn(\cen_i)$. Using $v$, we
  obtain a $d$-dimensional hyperplane $\hplane_i$ which is tangent to $\body_i$
  at $\cen_i$. Let $\hplane_i^+$ be the halfspace formed from $\hplane_i$ which
  contains the interior of $\body_i$. If $\cardin{\hplane_i^- \cap \PSet_i} \geq
  c\cardin{\PSet_i}/d^2$, then such an iteration is successful and we set
  $\PSet_{i+1} = \PSet_i \setminus (\hplane_i^- \cap \PSet_i)$ and continue to
  iteration $i+1$. Otherwise the iteration has failed and we repeat the $i$\th
  iteration. This procedure is repeated until we reach an iteration $\SuccIters$
  in which $\cardin{\PSet_\SuccIters}$ is of constant size. At this stage, the
  algorithm returns the point achieving the minimum of 
  $\min_{1 \leq i \leq \SuccIters} \fn(\cen_i)$ and 
  $\min_{\pnt \in \PSet_\SuccIters} \fn(\pnt)$. Because $\fn$ is convex, the 
  algorithm returns a point $\pntq$ such that $\fn(\pntq) \leq \fn(\pnt)$ for 
  all $\pnt \in \PSet$.

  As for the number of queries, note that in each iteration the expected number
  of centerpoint calculations (and thus queries) until a successful iteration is
  $O(1)$. It remains to bound the number of successful iterations. In each
  successful iteration, a $c/d^2$-fraction of points are discarded. Therefore
  there are at most $\SuccIters$ iterations, for which $\SuccIters$ is the
  smallest number with $(1 - c/d^2)^\SuccIters n$ smaller than some constant.
  This implies $\SuccIters = O(d^2 \log n)$.
\end{proof}

\subsection{Functional nets: A weak net in the oracle model}
\seclab{am:i:fat}

\subsubsection{The model, construction, and query process}

\paragraph*{Model.}
Given a convex body $\body \subseteq \Re^d$, we assume oracle
access to it. This is a standard model in optimization. Specifically,
given a query point $\query$, the oracle either returns that
$\query \in \body$, or alternatively it returns a (separating)
hyperplane $h$, such that $\body$ lies completely on one side of $h$,
and $\query$ lies on the other side.

Our purpose here is to precompute a small subset
$\Net \subseteq \PSet$, such that given any convex body $\body$ (with
oracle access to it), one can decide if $\body$ is
$\eps$-light. Specifically, the query algorithm (using only $\Net$,
and not the whole point set $\PSet$) generates an (adaptive) sequence
of query points $\query_1, \query_2, \ldots$, such that if any of
these query points are in $\body$, then the algorithm considers
$\body$ to be heavy. Otherwise, if all the query points miss $\body$,
then the algorithm outputs (correctly) that $\body$ is light (i.e.,
$\Measure{\body} < \eps$).

\paragraph*{Construction.}

Given $\PSet$, the set $\Net$ is a random sample from $\PSet$ of size
\begin{align}
  \SampleSize%
  =%
  O \pth{  \eps^{-1}  d^3 \log d  \log^3 \eps^{-1} +
  \eps^{-1} \log{\BadProb}^{-1}}%
  =%
  \tldO\pth{  d^3/\eps },
  \eqlab{net:size}%
\end{align}
where $\BadProb>0$ is a prespecified parameter.

\paragraph*{Query process.}

Given a convex body $\body$ (with oracle access to it), the algorithm
starts with $\Net_0 = \Net$. In the $i$\th iteration, the algorithm
computes a $\Omega(1/d^2)$-centerpoint $\query_i$ of $\Net_i$ using the
algorithm of \thmref{center:point:compute}, with failure probability
at most $1/4$.  If the oracle returns that $\query_i \in \body$, then
the algorithm returns $\query_i$ as a proof of why $\body$ is
considered to be heavy.  Otherwise, the oracle returns a separating
hyperplane $h_i$, such that the open halfspace $h_i^-$ contains
$\query_i$. Let $\Net_i' = \Net_{i-1} \setminus h_i^-$.  If
$\cardin{\Net_i'} \leq (1-\gamma) \cardin{\Net_{i-1}}$, where
$\gamma = 1/{16 d^2}$ then we set $\Net_i = \Net_{i}'$ (such an
iteration is called \emphi{successful}). Otherwise, we set
$\Net_i = \Net_{i-1}$.  The algorithm stops when
$\cardin{\Net_i} \leq \eps \cardin{\Net}/8$.

\subsubsection{Correctness}
\seclab{correctness}

Let $I$ be the set of indices of all the successful iterations, and
consider the convex set $\body_I = \cap_{i \in I} h_i^+$. The set
$\body_I$ is an outer approximation to $\body$. In particular, for an
index $j$, let $\body_j = \cap_{i \in I, i \leq j} h_i^+$ be this
outer approximation in the end of the $j$\th iteration.  We have that
$\Net_j = \Net \cap \body_j$.

\begin{lemma}
    \lemlab{tau:value}%
    There are at most $\SuccIters = O( d^2 \log \eps^{-1})$ successful
    iterations.  For any $j$, the convex polyhedron $\body_j$ is
    defined by the intersection of at most $\SuccIters$ closed
    halfspaces.
\end{lemma}
\begin{proof}
    We start with $\SampleSize = \cardin{\Net_0}$ points in
    $\Net_0$. Every successful iteration reduces the number of points
    in the net $\Net_{j-1}$ by a factor of $\gamma$. Furthermore,
    the algorithm stops as soon as
    $\cardin{\Net_j} \leq \eps \cardin{\Net_0}/8$. This implies there are
    at most $\SuccIters$ iterations, for the minimal $\SuccIters$ such
    that $(1-\gamma)^\SuccIters \leq \eps/8$, where
    $\gamma = {1}/{(16 d^2)}$. That is
    $\SuccIters = O( d^2 \log \eps^{-1})$.

    The second claim is immediate -- every successful iteration adds
    one halfspace to the intersection that forms $\body_I$.
\end{proof}

Let $\HI$ be the set of all of convex polyhedra in $\Re^d$ that are
formed by the intersection of $\SuccIters$ closed halfspaces.
\begin{observation}
    The \VC dimension of $(\Re^d, \HI)$ is
    \begin{align*}
      D%
      =%
      O( d \SuccIters \log \SuccIters)%
      =%
      O\pth{d \pth{ d^2 \log \eps^{-1}} \log
      \pth{ d^2 \log \eps^{-1}}}%
      =%
      O( d^3 (\log d)\log^2 \eps^{-1} ).
    \end{align*}
    This follows readily, as the \VC dimension of the range space of
    points in $\Re^d$ and halfspaces is $d+1$, and by the bound of
    \lemref{mixing:range:spaces} for the intersection of $\tau$ such
    ranges.
\end{observation}

\begin{lemma}
    \lemlab{relative:apx:calc}%
    The set $\Net$ is a relative $(\eps/8,1/4)$-approximation for
    $(\PSet, \HI)$, with probability $1 -\BadProb$.
\end{lemma}

\begin{proof}
    Using \thmref{relative} with $p=\eps/8$, $\epsilonA = 1/4$, and
    $\Dim = D$, implies that a random sample of $\PSet$ of size
    \begin{align*}
      O \pth{ \epsilonA^{-2}p^{-1}\pth{\Dim \log p^{-1} +
      \log{\BadProb}^{-1}}}%
      =%
      O \pth{  \eps^{-1} \pth{D \log \eps^{-1} +
      \log{\BadProb}^{-1}}}%
      =%
      O \pth{  \eps^{-1} \pth{d^3 \log d  \log^3 \eps^{-1} +
      \log{\BadProb}^{-1}}}%
    \end{align*}
    is the desired relative $(p,\epsilonA\ts)$-approximation with
    probability $\geq 1-\BadProb$. And this is indeed the size of $\Net$,
    see \Eqref{net:size}.
\end{proof}

\begin{lemma}
    Given a convex query body $\body$, the expected number of oracle
    queries performed by the algorithm is $O( d^2 \log \eps^{-1})$,
    and the expected running time of the algorithm is
    $O( d^{9} \eps^{-1} \polylog )$, where
    $\polylog = O\pth{ \log (d \eps^{-1}\BadProb^{-1})^{O(1)}}$.
\end{lemma}

\begin{proof}
    If the computed point is in the $i$\th iteration is indeed a
    centerpoint of $\Net_{i-1}$, then the algorithm would either stop
    in this iteration, or the iteration would be successful. Since the
    probability of the computed point to be the desired centerpoint is
    at least $\geq 3/4$, it follows that the algorithm makes (in
    expectation) $\tau/ (3/4)$ iterations till success. The $i$\th
    iteration requires $O( \cardin{\Net_i} + d^7 \log^3 d)$ time,
    since we use the algorithm of \thmref{center:point:compute} to
    compute the approximate centerpoint. Summing this over all the
    iterations $\tau$ (bounded in \lemref{tau:value}), we get expected
    running time
    \begin{align}
      O( d^2 \SampleSize + \SuccIters d^7 \log^3 d)%
      &=%
        O\pth{d^2    \pth{  \eps^{-1}  d^3 \log d  \log^3 \eps^{-1} +
        \eps^{-1} \log{\BadProb}^{-1}} + d^9 \log^3 d \log \eps^{-1}%
        }
        \nonumber
      \\&%
      =%
      O \pth{ d^2 \eps^{-1} \log{\BadProb}^{-1}
      + 
      d^5   \eps^{-1} \log d \log^3 \eps^{-1}
      +
      d^{9}\log^3 d
      \log\eps^{-1} }. ~
      \eqlab{exact:r:t}%
    \end{align}    
\end{proof}

\begin{lemma}
    Assuming that $\Net$ is the desired relative approximation, then
    for any query body $\body$, if the algorithm declares that it is
    $\eps$-light, then $\cardin{\body \cap \PSet} < \eps n$.
\end{lemma}
\begin{proof}
    Let $I$ be the set of successful iterations by the algorithm, 
    and recall that $\body_I$ is the outer approximation of $\body$ 
    and is the intersection of at most $\SuccIters$ halfspaces.
    Now the algorithm stops in an iteration $i$ when 
    $\cardin{\Net_{i}} \leq (\eps/8)\cardin{\Net}$. Consequently,
    $\sMeasure{\body} \leq \sMeasure{\body_I} = 
      \cardin{\Net_i}/\cardin{\Net} \leq \eps/8$.
    There are two cases:
    \begin{compactenumi}
      \item If $\Measure{\body_I} < \eps/8$, then $\Measure{\body} \leq 
            \Measure{\body_I} < \eps$ as claimed.
      \item Otherwise $\Measure{\body_I} \geq \eps/8$. But then, since
            $\Net$ is a relative $(\eps/8,1/4)$-approximation for the 
            points $\PSet$ and ranges $\HI$ and $\body_I \in \HI$, 
            we have that
            $\Measure{\body} \leq \Measure{\body_I} 
            \leq \frac{1}{1 - 1/4} \sMeasure{\body_I} 
            \leq (4/3)(\eps/8) < \eps$.
    \end{compactenumi}
    In either case, the algorithm is correct.
\end{proof}

The above implies the following.

\begin{theorem}
    \thmlab{func:net}%
    Let $\PSet$ be a set of points in $\Re^d$, and let
    $\eps, \BadProb > 0 $ be parameters. Let $\Net$ be a random sample
    of $\PSet$ of size
    \begin{align*}
      \SampleSize%
      =%
      O \pth{  \eps^{-1}  d^3 \log d  \log^3 \eps^{-1} +
      \eps^{-1} \log{\BadProb}^{-1}}%
      =%
      \tldO( \eps^{-1} d^3).
    \end{align*}

    Then, for a given query convex body $\body$ endowed with an oracle
    access, the algorithm described above, which uses only $\Net$,
    computes a sequence of query points $q_1, \ldots, q_m$, such
    that either:
    \begin{compactenumi}
        \item one of the points $q_i \in \body$, and the
        algorithm outputs $q_i$ as a ``proof'' that $\body$ is
        $\eps$-heavy, or
        
        \item the algorithm outputs that
        $\cardin{\body \cap \PSet} < \eps n$.
    \end{compactenumi}
    The query algorithm has the following performance guarantees:
    \begin{compactenumA}
        \item The expected number of oracle queries is
        $\Ex{m} = O(d^2 \log \eps^{-1})$.

        \item The algorithm itself (ignoring the oracle queries) runs
        in $\tldO\pth{ d^9 \eps^{-1}}$ time (see \Eqref{exact:r:t} for
        exact bound).
    \end{compactenumA}
        
    The output of the algorithm is correct, for all convex bodies, 
    with probability $\geq 1 - \BadProb$.
\end{theorem}

\begin{remark}
    One may hope to bound the probability of the algorithm
    reporting a false positive. However this is inherently not
    possible for any weak $\eps$-net construction. Indeed, the
    algorithm can fail to distinguish between a polygon that contains
    at least $\eps n$ of the points of $\PSet$ and a polygon that
    contains \emph{none} of the points of $\PSet$. Consider
    $n$ points $\PSet$ lying on a circle in $\Re^2$. Choose $\eps n$
    of these points on the circle, and let $\body$ be the convex hull
    of these points. Clearly $\body$ contains at least $\eps n$ points
    of $\PSet$. Now, take each vertex in $\body$ and ``slice'' it off,
    forming a new polygon $\body'$ that contains no points from
    $\PSet$. However, $\body'$ is still a large polygon and as such
    may contain a centerpoint during the execution of the above
    algorithm. Therefore our algorithm may report that $\body'$
    contains a large fraction of the points, even though $\body'$ is
    contains no points of $\PSet$, and so it fails to distinguish
    between $\body$ and $\body'$.
\end{remark}

\begin{remark}
    \remlab{better:center:slower}%
    Clarkson \etal \cite{cemst-acpir-96} provide also a randomized
    algorithm that finds a
    $\bigl(\frac{1}{d+1}-\gamma\bigr)$-centerpoint with probability
    $1 - \delta$ in time
    \begin{math}
        O\bigl( \bigl[d\gamma^{-2} \log(d\gamma^{-1})\bigr]^{d + O(1)}
        \log\delta^{-1}\bigr).
    \end{math}
    We could use this algorithm instead of
    \thmref{center:point:compute} in the query process. Since we are
    computing a better quality centerpoint, the number of iterations
    $\tau$ and sample size $\SampleSize$ would be smaller by a factor
    of $d$.  Specifically, $\tau = O(d \log \eps^{-1})$ and from
    \lemref{mixing:range:spaces}, the $\VC$ dimension of the range
    space $\RangeSpace = (\PSet, \HI)$ becomes
    $D = O(d^2 \log d \log^2 \eps^{-1})$. Following the proof of
    \lemref{relative:apx:calc}, we can construct a sample $\Net$ which
    is $(\eps/8, 1/4)$-relative approximation for $\RangeSpace$ with
    probability $1 - \BadProb$ of size
    \begin{align}
      \SampleSize%
      =%
      O\pth{\eps^{-1}(D \log \eps^{-1} + \log \BadProb^{-1})}
      =
      O\pth{\eps^{-1}(d^2 \log d \log^3 \eps^{-1} + \log \BadProb^{-1})}.
      \eqlab{exact:bound:2}
    \end{align}
\end{remark}


\section{Application II: Constructing center nets}
\seclab{center:nets}%

We next introduce a strengthening of the concept of a weak
$\eps$-net. Namely, we require that there is a point $\pnt$ in the net
which stabs an $\eps$-heavy convex body $\body$, and that $\pnt$ is
also a good centerpoint for $\body \cap \PSet$.

\begin{definition}
    \deflab{center:net}%
    For a set $\PSet$ of $n$ points in $\Re^d$, and parameters $\eps,
    \alpha \in (0,1)$, a subset $\Wnet \subseteq \Re^d$ is an
    \emphi{$(\eps,\alpha)$-center net} if for any convex shape
    $\body$, such that $\cardin{\PSet \cap \body } \geq \eps n$, we
    have that there is an $\alpha$-centerpoint of $\PSet \cap \body$
    in $\Wnet$.
\end{definition}    

In this section we prove existence of an $(\eps, \cpq)$-net $\Wnet$
of size roughly $O_d(\eps^{-d^2})$, where
\begin{align*}
  \cpq = \frac{\constCN}{(d+1)\log \eps^{-1}},
\end{align*}
and $\constCN \in (0,1)$ is some fixed constant to be specified
shortly. Note that the quality of the centerpoint is worse by a
factor of $\log \eps^{-1}$ than the best one can hope for. 

\subsection{The construction}
\seclab{w:eps:net:c}%

The construction of the center net will be based an algorithm for constructing a
weak $\eps$-net for $\PSet$. In particular, the construction algorithm will use
the following two results (see \apndref{weak:eps:net} for the proofs).

\SaveContent{\LemmaUCP}{
    Given a set $\PSet$ of $n$ points in $\Re^d$, one can compute a
    set $\PSetA$ of $O\bigl( n^{d^2} \bigr)$ points, such that for any
    subset $\PSet' \subseteq \PSet$, there is a $1/(d+1)$-centerpoint
    of $\PSet'$ in $\PSetA$.
}

\begin{lemma}
    \lemlab{u:c:p}
    \LemmaUCP
\end{lemma}

\SaveContent{\LemmaWeakNetConstruction}{
  Let $\PSet$ be a set of $n$ points in $\Re^d$. Let $\Net$ be a
  random sample from $\PSet$ of size
  $\SampleSize = \tldO(\eps^{-1} d^2)$, see \Eqref{exact:bound:2}
  for the exact bound. Then, one can compute a set of points
  $\Wnet$ from $\Net$, of size
  \begin{align*}
    O(\SampleSize^{d^2})%
    =%
    O\pth{\pth{\eps^{-1}(d^2 
    \log d \log^3 \eps^{-1} + \log \BadProb^{-1})}^{d^2}}
  \end{align*}
  which is a weak $\eps$-net for $\PSet$ with probability
  $\geq 1 - \BadProb$.
}

\begin{lemma}
    \lemlab{weak:net:construction}%
    \LemmaWeakNetConstruction
\end{lemma}

\begin{remark}
    A similar construction of a weak $\eps$-net, to the one in
    \lemref{weak:net:construction}, from a small sample was described
    by Mustafa and Ray \cite{mr-wenbs-08}. Their sample has
    exponential dependency on the dimension, so the resulting weak
    $\eps$-net has somewhat worse dependency on the dimension than our
    construction. In any case, these constructions are inferior as far
    as the dependency on $\eps$, compared to the work of \Matousek and
    Wagner \cite{mw-ncwen-04}, which shows a weak $\eps$-net
    construction of size
    $O_d( \eps^{-d} (\log \eps^{-1})^{O( d^2 \log d)} )$. %
\end{remark} %

The idea will be to repeat the construction of the net of
\lemref{weak:net:construction}, with somewhat worse constants.
Specifically, take a sample $\Net$ of size $\SampleSize = \tldO(\eps^{-1} d^2)$
from $\PSet$, see \Eqref{exact:bound:2} for the exact bound. Next, we construct
the set $\Wnet$ for $\Net$, using the result of \lemref{u:c:p}.  Return
$\Wnet$ as the desired $(\eps, \cpq)$-center net.


\subsection{Correctness}
\seclab{c:n:correctness}

The proof is algorithmic.  Fix any convex $\eps$-heavy body $\body$,
and let $\Net_1 = \Net$ be the \emphi{active set} and let
$\PSet_1 = \body \cap \PSet$ be the \emphi{residual set} in the
beginning of the first iteration.

We now continue in a similar fashion to the
algorithm of \thmref{func:net}.  In the $i$\th iteration, the
algorithm computes the $1/(d+1)$-centerpoint $\query_i$ of $\Net_i$
(running times do not matter here, so one can afford computing the
best possible centerpoint).  If $\query_i$ is a $2\cpq$-centerpoint
for $\PSet_i$, then $\query_i$ is intuitively a good centerpoint for
$\PSet$, and the algorithm returns $\query_i$ as the desired center
point. Observe that by construction, $\query_i \in \Wnet$ as
desired.

If not, then there exists a closed halfspace $h_i^+$ containing
$\query_i$ and  at most $2\cpq \cardin{\PSet_i}$ points of
$\PSet_i$. Let
\begin{align*}
  \PSet_{i+1} = \PSet_i \setminus h_i^+
  \qquad\text{ and }\qquad%
  \Net_{i+1} = \Net_i \setminus h_i^+.
\end{align*}
The algorithm now continues to the next iteration. 

\paragraph*{Analysis.}
The key insight is that the active set $\Net_i$ shrinks much faster
than the residual set $\PSet_i$. However, by construction, $\Net_i$
provides a good upper bound to the size of $\PSet$. Now once the
upper bound provided by $\Net_i$ on the size of $\PSet_i$ is too
small, this would imply that the algorithm must have stopped earlier,
and found a good centerpoint.

The proof of the following two lemmas are in 
\apndref{missing:proofs:center:net}.

\SaveContent{\LemmaCenterNetIters}{
    Let %
    \begin{math} %
        \SuccIters = \ceil{1+ 3(d+1) + (d+1) \log\eps^{-1}},
    \end{math}
    and
    \begin{math}
        \cpq = 1/(4\SuccIters).
    \end{math}
    Assuming that $\Net$ is a relative
    $(\eps/8, 1/4)$-approximation for the range space
    $\RangeSpace = (\PSet, \HI)$, the above algorithm stops after
    at most $\SuccIters$ iterations.
}

\begin{lemma}
    \lemlab{center:net:iters}
    \LemmaCenterNetIters
\end{lemma}

\SaveContent{\LemmaCenterNetItersProof}{
    As before, we can interpret the algorithm as constructing a convex
    polyhedra. Indeed, let $\bodyA_{i+1} = \bigcup_{j=1}^{i} h_j^-$, and
    observe that $\PSet_{i+1} \subseteq \PSet \cap \bodyA_{i+1} $, and
    $\Net_{i+1} = \Net \cap \bodyA_{i+1}$. %

    For an iteration $i < \SuccIters$, we have
    \begin{align*}
        n_{i+1}%
        =%
        \cardin{\PSet_{i+1}}%
        \geq%
        (1 - 2\cpq) \cardin{\PSet_i}%
        \geq%
        (1 - 2\cpq)^i \cardin{\PSet_1} %
        \geq%
        \pth{1 -2\cpq i} n_1 \geq%
        \pth{1 -2\cpq \SuccIters} n_1 %
        \geq%
        (\eps/2) n,
    \end{align*}
    using $(1-x)^i \geq 1-ix$, which holds for any positive $x \in [0,1]$.
    
    On the other hand, the active set shrinks faster in each such
    iteration, since $\query_{i}$ is a $1/(d+1)$-centerpoint of
    $\Net_i$. Setting $s_i = \cardin{\Net_i}$, we have that
    \begin{align*}
      s_{i+1}%
      \leq%
      \pth{1 -\frac{1}{d+1}} s_i%
      \leq%
      \pth{1 -\frac{1}{d+1}}^i s_1%
      \leq%
      \exp \pth{ -\frac{i}{d+1}} s_1.
    \end{align*}
    We have that
    \begin{math}
        s_{\SuccIters} \leq \eps s_1 /e^3 \leq (\eps/20)s_1.
    \end{math}
    It follows that,
    \begin{align*}
      \frac{\eps}{2}%
      &\leq%
        \frac{\cardin{\PSet_{\SuccIters}}}{n}%
        \leq%
        \frac{\cardin{\PSet \cap \bodyA_{\SuccIters}}}{n}%
        =%
        \Measure{\bodyA_{\SuccIters}}  %
        \leq %
        \max\pth{\frac{1}{1-1/4}
        \sMeasure{\bodyA_{\SuccIters}},
        \frac{\eps}{8}\cdot \frac{1}{4} + \sMeasure{\bodyA_{\SuccIters}}}
        \leq \frac{7}{3}\sMeasure{\bodyA_{\SuccIters}} + \frac{\eps}{32}.
    \end{align*}
    The penultimate inequality follows since
    $S$ is a relative $(\eps/8, 1/4)$-approximation to $\PSet$ for 
    ranges like $\bodyA_\SuccIters$. However we do not know which case
    applies (i.e., depending on whether or not 
    $\Measure{\bodyA_\SuccIters} \geq \eps/8$)
    and therefore need to take the maximum over both cases.
    Finally,
    \begin{align*}
      \frac{\eps}{2}
        &\leq 
        \frac{7}{3} \sMeasure{\bodyA_{\SuccIters}}  + \frac{\eps}{32}
        = %
        \frac{7}{3} \frac{\cardin{\Net_{\SuccIters}}}{\cardin{\Net}}
        +
        \frac{\eps}{32}
      \leq %
      \frac{7}{3} \frac{(\eps/20) s_1}{s_1}  + \frac{\eps}{32}
      <
      \frac{\eps}{5},
    \end{align*}
    which is impossible. We conclude the algorithm must have stopped at an
    earlier iteration.
}

\SaveContent{\LemmaCenterNetCorrect}{
    The above algorithm outputs a $\cpq$-centerpoint of $\PSet \cap
    \body$.
}

\begin{lemma}
    \lemlab{center:net:correct}
    \LemmaCenterNetCorrect
\end{lemma}

\SaveContent{\LemmaCenterNetCorrectProof}{
    Assume the algorithm stopped in the $i$\th iteration. But then
    $\query_i$ is a $2\cpq$-centerpoint of $\PSet_i$. Since
    $n_i \geq n_\SuccIters \geq n_1/2$, it follows that any closed
    halfspace that contains $\query_i$, contains at least
    $2\cpq n_i \geq \cpq n_1$ points of $\PSet_i$, and thus of
    $\PSet_1$. We conclude that $\query_i$ is a $\cpq$-centerpoint of
    $\PSet$ as desired.
}

Arguing as in \remref{better:center:slower} implies the following.

\begin{corollary}
    For the above algorithm to succeed with probability
    $\geq 1 -\BadProb$, the sample $\Net$ needs to be a sample of the
    size specified by \Eqref{exact:bound:2}.
\end{corollary}

\subsection{The result}

\begin{theorem}
    \thmlab{center:net}%
    Let $\PSet$ be a set of $n$ points in $\Re^d$, and $\eps > 0$ be a
    parameter. For $\gamma = \log(1/\eps)$, there exists a
    $\bigl(\eps, \Omega(1/( d\gamma)) \bigr)$-center net $\Wnet$ (which is
    also a weak $\eps$-net) of $\PSet$ (see \defref{center:net}). The
    size of the net $\Wnet$ is
    $O( \SampleSize^{d^2}) \approx O_d( \eps^{-d^2})$, where
    $\SampleSize = \tldO(\eps^{-1} d^2)$, see \Eqref{exact:bound:2}
    for the exact bound.
\end{theorem}
\begin{proof}
    The theorem follows readily from the above, by setting
    $\BadProb = 1/2$.
\end{proof}




\BibTexMode{%
\newcommand{\etalchar}[1]{$^{#1}$}
 \providecommand{\CNFX}[1]{ {\em{\textrm{(#1)}}}}
  \providecommand{\CNFCCCG}{\CNFX{CCCG}}

}

\BibLatexMode{\printbibliography}

\begin{thebibliography}{CEM{\etalchar{+}}96}

\bibitem[CEM{\etalchar{+}}96]{cemst-acpir-96}
\hrefb{http://cm.bell-labs.com/who/clarkson/}{K.~L. {Clarkson}}, \hrefb{http://www.ics.uci.edu/~eppstein/}{D.~{Eppstein}}, G.~L. Miller, C.~Sturtivant, and S.-H. Teng.
\newblock Approximating center points with iterative {Radon} points.
\newblock {\em Internat. J. Comput. Geom. Appl.}, 6:357--377, 1996.

\bibitem[{Har}11]{h-gaa-11}
\hrefb{http://sarielhp.org}{S.~{{Har-Peled}}}.
\newblock {\em Geometric Approximation Algorithms}, volume 173 of {\em Math.
  Surveys \& Monographs}.
\newblock Amer. Math. Soc., Boston, MA, USA, 2011.

\bibitem[HS11]{hs-rag-11}
\hrefb{http://sarielhp.org}{S.~{{Har-Peled}}} and \hrefb{http://www.math.tau.ac.il/~michas}{M.~{Sharir}}.
\newblock Relative {$(p,\varepsilon)$}-approximations in geometry.
\newblock {\em \hrefb{http://link.springer.com/journal/454}{Discrete Comput. {}Geom.}}, 45(3):462--496, 2011.

\bibitem[HW87]{hw-ensrq-87}
D.~Haussler and E.~Welzl.
\newblock {$\varepsilon$}-nets and simplex range queries.
\newblock {\em \hrefb{http://link.springer.com/journal/454}{Discrete Comput. {}Geom.}}, 2:127--151, 1987.

\bibitem[{Jan}17]{j-tbsge-17}
S.~{Janson}.
\newblock {Tail bounds for sums of geometric and exponential variables}.
\newblock {\em ArXiv e-prints}, September 2017.

\bibitem[KST15]{kst-omcig-15}
C.~{Keller}, S.~{Smorodinsky}, and G.~{Tardos}.
\newblock {Improved bounds on the Hadwiger-Debrunner numbers}.
\newblock {\em ArXiv e-prints}, December 2015.
\newblock \url{https://arxiv.org/abs/1512.04026}.

\bibitem[KST17]{kst-omcig-17}
C.~Keller, S.~Smorodinsky, and G.~Tardos.
\newblock On {Max-Clique} for intersection graphs of sets and the
  hadwiger-debrunner numbers.
\newblock In Philip~N. Klein, editor, {\em Proc. 28th ACM-SIAM Sympos. Discrete
  Algs. {\em(SODA)}}, pages 2254--2263. {SIAM}, 2017.

\bibitem[LLS01]{lls-ibscl-01}
Y.~Li, P.~M. Long, and A.~Srinivasan.
\newblock Improved bounds on the sample complexity of learning.
\newblock {\em J. Comput. Syst. Sci.}, 62(3):516--527, 2001.

\bibitem[Mat02]{m-ldg-02}
\hrefb{http://kam.mff.cuni.cz/~matousek}{J. Matou{\v s}ek}.
\newblock {\em Lectures on Discrete Geometry}, volume 212 of {\em Grad. Text in
  Math.}
\newblock Springer, 2002.

\bibitem[MR08]{mr-wenbs-08}
N.~H. Mustafa and S.~Ray.
\newblock Weak {$\varepsilon$}-nets have basis of size {$O(\varepsilon^{-1}
  \log \varepsilon^{-1})$} in any dimension.
\newblock {\em Comput. Geom. Theory Appl.}, 40(1):84--91, 2008.

\bibitem[MS10]{ms-acwp-10}
G.~L. Miller and D.~R. Sheehy.
\newblock Approximate centerpoints with proofs.
\newblock {\em Comput. Geom.}, 43(8):647--654, 2010.

\bibitem[MV17]{mv-eaen-17}
N.~H. Mustafa and K.~Varadarajan.
\newblock Epsilon-approximations and epsilon-nets.
\newblock {\em CoRR}, abs/1702.03676, 2017.

\bibitem[MW04]{mw-ncwen-04}
J.~Matou{\v{s}}ek and U.~Wagner.
\newblock New constructions of weak epsilon-nets.
\newblock {\em \hrefb{http://link.springer.com/journal/454}{Discrete Comput. {}Geom.}}, 32(2):195--206, 2004.

\bibitem[RS16]{rs-wnmp-16}
A.~Rok and S.~Smorodinsky.
\newblock Weak {$1/r$}-nets for moving points.
\newblock In {\em Proc. 32nd Int. Annu. Sympos. Comput. Geom. {\em(SoCG)}},
  pages 59:1--59:13, 2016.

\bibitem[Rub18]{r-ibwpnp-18}
N.~Rubin.
\newblock An improved bound for weak epsilon-nets in the plane.
\newblock In {\em Proc. 59th Annu. Sympos. on Found. of Comput. Sci.,
  {\em(FOCS)}}, pages 224--235, 2018.

\bibitem[VC71]{vc-ucrfe-71}
V.~N. Vapnik and A.~Y. Chervonenkis.
\newblock On the uniform convergence of relative frequencies of events to their
  probabilities.
\newblock {\em Theory Probab. Appl.}, 16:264--280, 1971.

\end{thebibliography}


\appendix

\section{Missing proofs}

\subsection{Missing proofs from \secref{approx:centerpoint}}
\apndlab{missing:proofs:approx:centerpoint}

Consider the random walk that starts at $Y_0 = (1-\epsA)\rmax$. At the
$i$\th iteration, $Y_i = Y_{i-1} -1$ with probability $1/2+\epsW$, and
$Y_i = Y_{i-1} +1$ with probability $1/2 - \epsW$. Let
$\YY = Y_1, Y_2, \ldots$ be the resulting random walk which
stochastically dominates the walk $X_0, X_1, \ldots$.


\subsubsection{Proof of \lemref{good:range}}
\apndlab{good:range:proof}

\RestatementOf{\lemref{good:range}}{\LemmaGoodRange}
\begin{proof}
  \LemmaGoodRangeProof
\end{proof}

\subsubsection{Proof of a generalized version of %
   \lemref{number:visits}}
\apndlab{number:visits:proof}%

\RestatementOfExt{A more general
   version}{\lemref{number:visits}}{\LemmaNumberVisits}
\begin{proof}
  \LemmaNumberVisitsProof
\end{proof}

\subsubsection{Proof of \lemref{H:bound}}
\apndlab{H:bound:proof}

\RestatementOf{\lemref{H:bound}}{\LemmaHBound}
\begin{proof}
  \LemmaHBoundProof
\end{proof}


\subsubsection{Proof of \lemref{stuck:on:r}}
\apndlab{stuck:on:r:proof}

\RestatementOf{\lemref{stuck:on:r}}{\LemmaStuckOnR}
\begin{proof}
  \LemmaStuckOnRProof  
\end{proof}

In the above we used concentration inequality for sum of geometric
variables. It is easy to derive such inequality from Chernoff type
inequalities, see \cite[Theorem 2.3]{j-tbsge-17}.

\begin{corollary}
    \corlab{geometric:sum}%
    Let $X_1, \ldots, X_m$ be $m$ independent random variables with
    geometric distribution with parameter $1/2$. For any
    $\lambda \geq 1$, we have
    \begin{math}
        \Prob{ \sum_{i=1}^m X_i \geq \lambda \cdot 2 m}%
        \leq%
        \lambda^{-1} 2^{-2m (\lambda -1 - \ln \lambda)}.
   \end{math}
   For $\lambda = 4$, this becomes
    \begin{math}
      \Prob{ \sum_{i=1}^m X_i \geq 8m }%
      \leq%
      2^{-2m (4 -1 - \ln 4)} / 4
      \leq%
      2^{-3m},
  \end{math}
  for $m \geq 10$.
\end{corollary}


\subsubsection{Proof of \lemref{iterations}}
\apndlab{iterations:proof}

\RestatementOf{\lemref{iterations}}{\LemmaIterations}
\begin{proof}
  \LemmaIterationsProof  
\end{proof}


\subsubsection{Proof of a more general version of \lemref{center}}
\apndlab{center:proof}

\RestatementOf{\lemref{center}}{\LemmaCenter}
\begin{proof}
  Set $\epsA = \epsB/2$ and $\epsW = \epsB/4$.
  The proof follows exactly from the discussion in 
  \secref{centerpoint:analysis} and \lemref{iterations}
  with the introduced parameter $\epsW$.
  The key change is in \Eqref{f:d}. Since
  $\rmax = (1 - 2\epsW)n/(d+2)^2$ and we require that
  $(1 - \epsA)R = f(d)n$, we obtain
  \begin{align*}
    f(d) 
    \geq \frac{(1 - \epsA)(1 - 2\epsW)}{(d+2)^2}
    = \frac{(1 - \epsB/2)^2}{(d+2)^2}
    \geq \frac{1 - \epsB}{(d+2)^2}.
  \end{align*}
\end{proof}


\subsubsection{Proof of \thmref{center:point:compute:epsW}}
\apndlab{center:point:compute:proof:eps:W}

\RestatementOf{\thmref{center:point:compute:epsW}}{\ThmCenterPointCompute}
\begin{proof}
    Pick a random sample $\Net$ from $\PSet$ that is a
    $(\rho,\epsB/2)$-relative approximation for halfspaces, where
    $\rho = 1/(8d^2)$. This range space has \VC dimension $d+1$, and
    by \thmref{relative}, a sample of size
    \begin{align*}
      \SampleSize %
      =%
      O\bigl(\rho^{-1} \epsB^{-2}(d \log \rho^{-1} + \log
      \BadProb^{-1})\bigr)%
      =%
      O\bigl(d^2\epsB^{-2} ( d \log d + \log \BadProb^{-1})\bigr)
    \end{align*}
    is a $(\rho,\epsB/2)$-relative approximation.

    Running the algorithm of \lemref{center} on $\Net$ with
    $\epsB/2$ yields a $\tau$-centerpoint $\cPnt$ of $\Net$, where
    \begin{math}
        \tau= \frac{1-\epsB/2}{(d+2)^2\bigr.} \geq \rho
    \end{math}
    for $d \geq 2$.
    By the relative approximation property, this is a
    $(1\pm \epsB/2)\tau$-centerpoint of $\PSet$. As such, we have
    that $\cPnt$ is an $\cpq$-centerpoint for $\PSet$, where
    \begin{align*}
      \cpq %
      =%
      (1-\epsB/2) \tau %
      =%
      \frac{(1-\epsB/2)^2}{(d+2)^2}
      \geq%
      \frac{1-\epsB}{(d+2)^2}
    \end{align*}
    By \lemref{center}, the running time of the resulting algorithm is
    \begin{align*}
      O(\epsB^{-2} d^4 \SampleSize \log \SampleSize \log(
      \SampleSize \epsB^{-1} \BadProb^{-1}))%
      =%
      O\bigl( \epsB^{-4} d^7 \log^3 d \log^3 \pth{
      \epsB^{-1} \BadProb^{-1} }\bigr).
    \end{align*}
\end{proof}

\subsection{Missing proofs from \secref{w:eps:net:c}}
\apndlab{weak:eps:net}%



\subsubsection{Proof of \lemref{u:c:p}}
\apndlab{u:c:p:proof}

\RestatementOf{\lemref{u:c:p}}{\LemmaUCP}

\begin{proof}
    This is well known, and we include a proof for the sake of
    completeness.
    
    Let $\hset$ be the set of all hyperplanes which pass through
    $d$ points of $\PSet$. The original proof of the centerpoint
    theorem implies that a vertex of the arrangement $\ArrX{\hset}$
    is a $1/(d+1)$-centerpoint of $\PSet$.  Let $\VVX{\PSet}$ denote
    the set of vertices of $\ArrX{\hset}$. Observe that
    $\VVX{\PSet'} \subseteq \VVX{\PSet}$, for all
    $\PSet' \subseteq \PSet$, thus implying that $\VVX{\PSet}$
    contains all desired centerpoints. As for the size bound, observe
    that
    \begin{align*}
      \alpha = \cardin{\hset} \leq \binom{n}{d} \leq \pth{\frac{ne}{d}}^d %
    \end{align*}
    and
    \begin{align*}
      \cardin{\VVX{\PSet}} \leq \binom{\alpha}{d}
      \leq
      \pth{\frac{e \pth{\frac{ne}{d}}^d}{d}}^d%
      =%
      n^{d^2} \pth{\frac{e}{d}}^{d^2 + d}%
      = %
      O\bigl( n^{d^2} \bigr).
    \end{align*}
\end{proof}


\subsubsection{Proof of \lemref{weak:net:construction}}
\apndlab{weak:net:construction:proof}

\RestatementOf{\lemref{weak:net:construction}}{\LemmaWeakNetConstruction}

\begin{proof}
    Imagine running the algorithm of \thmref{func:net} with the better
    quality centerpoint algorithm, as sketched in
    \remref{better:center:slower}. This requires computing a sample
    $\Net$ of size as specified in \Eqref{exact:bound:2}. Let $\PSetA$
    be the universal set of centerpoints, as computed by
    \lemref{u:c:p}, for the set $\Net$. We claim that $\PSetA$ is a
    weak $\eps$-net. Indeed, assuming the sample $\Net$ is good, in
    the sense that the algorithm of \thmref{func:net} works (which
    happens with probability $\geq 1- \BadProb$), then running this
    algorithm on any convex query body $\body$ (using $\Net$),
    generates a sequence of points, such that one of them stabs
    $\body$, if $\body$ is $\eps$-heavy. However, the stabbing points
    computed by the algorithm of \thmref{func:net} are centerpoints
    of some subset of $\Net\!$.

    It follows that all the stabbing points that might be computed by
    the algorithm of \thmref{func:net}, over all possible $\eps$-heavy
    query bodies $\body$ are contained in the set $\PSetA$. Consequently,
    if $\body$ is $\eps$-heavy, then $\body$ must contain one of the
    points of $\PSetA$.
\end{proof}


\subsection{Missing proofs from \secref{c:n:correctness}}
\apndlab{missing:proofs:center:net}

\subsubsection{Proof of \lemref{center:net:iters}}
\apndlab{center:net:iters}

\RestatementOf{\lemref{center:net:iters}}{\LemmaCenterNetIters}
\begin{proof}
  \LemmaCenterNetItersProof
\end{proof}

\subsubsection{Proof of \lemref{center:net:correct}}
\apndlab{center:net:correct}

\RestatementOf{\lemref{center:net:correct}}{\LemmaCenterNetCorrect}
\begin{proof}
  \LemmaCenterNetCorrectProof
\end{proof}

\section{An easy lemma}
\apndlab{silly}

\begin{lemma}
    \lemlab{silly}%
    For any $p\in [0,1]$, and integers $0 \leq k \leq n$, we have
    \begin{math}
        \sum_{i=k}^n \binom{n}{i} p^i (1-p)^{n-i}%
        \leq \binom{n}{k} p^k.
    \end{math}
\end{lemma}
\begin{proof}
    Observe that for $k \leq i \leq n$, we have
    $\binom{n}{i} \leq \binom{n}{k}\binom{n-k}{i-k}$ -- indeed, the
    right side counts the number of ways to first choose $k$ elements
    of the set, and then choose additional $i-k$ elements from the
    remaining $n-k$ elements, which is definitely more than the number
    of ways to choose $i$ elements out of $n$ elements. Therefore,
    \begin{align*}
        \sum_{i=k}^n \binom{n}{i} p^i (1-p)^{n-i}%
        &=%
        \sum_{\ell=0}^{n-k} \binom{n}{\ell+k} p^{\ell+k} (1-p)^{(n-k)-\ell}\\%
        &\leq%
        \binom{n}{k}p^k\sum_{\ell=0}^{n-k} \binom{n-k}{\ell} p^{\ell}
        (1-p)^{(n-k)-\ell}%
        =%
        \binom{n}{k}p^k.
    \end{align*}
\end{proof}

\section{Animation of the centerpoint algorithm}
\apndlab{cpnt:anim}

\centerline{\animategraphics
   [autoplay,autoresume,width=0.8\linewidth,controls,loop,buttonsize=4em]%
   {2}{figs/cpnt_anim}{0}{19}}

\medskip
An animation of the centerpoint algorithm from
\thmref{center:point:compute}. Adobe Acrobat reader is needed in order
to view the animation properly. Initially, each of the three large
circles have 100 points and the three small circles have 400 points.

\end{document}